\documentclass[runningheads]{llncs}

\usepackage{microtype}
\usepackage{todonotes}
\usepackage{amsmath}
\usepackage{amsmath,amsfonts,amssymb} 
\usepackage{url}
\usepackage{stmaryrd}
\usepackage{hyperref}
\usepackage{enumerate}
\usepackage{todonotes}
\usepackage{wrapfig}
\usepackage{pgfplots}
\usepackage{multirow}
\usepackage{verbatim}
\usepackage{mathtools}

\usepackage{enumerate}
\usepackage{paralist}
\usepackage{extarrows}
\usepackage{mathtools}

\usepackage{wasysym}

\usepackage{algorithm}
\usepackage{algorithmic}
\usepackage{listings}
\usepackage[title]{appendix}
\usepackage{marginnote}
\usepackage{listings}
\usepackage{url}
\usepackage{wrapfig}
\usepackage{graphicx}
\usepackage{tikz}
\usepackage{xcolor,colortbl}
\usetikzlibrary{positioning,shadows.blur}
\usetikzlibrary{arrows,automata}
\usepackage{pifont}
\usepackage{hhline}
\usepackage{xspace}

\allowdisplaybreaks

\urldef{\mailsa}\path|{pedro.dargenio,rdemasi}@unc.edu.ar|
\urldef{\mailsb}\path|{pcastro,lputruele}@dc.exa.unrc.edu.ar|

\def\Rewards{\textit{rew}}

\def\Probabilistic{\mathsf{P}}

\def\pr#1#2{\mbox{pr}_{#1}(#2)}

\def\val{\mathit{val}}
\def\poly{\mathit{poly}}
\def\size{\mathit{size}}
\def\MP{\mathit{MP}}

\def\support#1{\mathit{Supp}(#1)}

\def\FP{\mathit{FP}}

\def\Prob#1#2{\mathit{Prob}^{#1, #2}}
\def\MDPProb#1{\mathit{Prob}^{#1}}

\def\GamePaths{\mathit{Paths}}
\def\FinGamePaths{\mathit{FPaths}}

\def\MDPPaths#1{\mathit{Paths}^{#1}}

\def\MCFinPaths#1{\mathit{FPaths}}
\def\MDP#1{\mathcal{#1}}
\def\MC#1{\mathcal{#1}}
\def\EC#1{\mathcal{#1}}
\def\EndComp{\mathit{EC}}

\def\pr#1#2{\operatorname{pr}_{#1}(#2)}
\DeclareMathOperator*{\argmin}{argmin}
\DeclareMathOperator*{\argmax}{argmax}

\def\Expect#1#2{\mathbb{E}^{#1, #2}}
\def\ExpectMDP#1{\mathbb{E}^{#1}}
\def\ExpectMC{\mathbb{E}}

\def\StochG{\mathcal{G}}

\def\post{\mathit{post}}
\def\pre{\mathit{pre}}
\def\Apre{\forall\mathit{Pre}}
\def\Epre{\exists\mathit{Pre}}
\def\AFairpre{\forall\mathit{Pre}_f}
\def\EFairpre{\exists\mathit{Pre}_f}

\def\out{\mathit{out}}
\def\Dist{\mathcal{D}}
\def\reward{\mathit{r}}

\def\cone{\mathit{cyl}}
\def\cylinder#1{\mathit{cyl}(#1)}

\def\strat#1{\pi_{#1}}

\def\starredstrat#1{\pi^*_{#1}}
\def\hatstrat#1{\hat{\pi}_{#1}}
\def\limit{\mathit{limit}}

\def\Strategies#1{\Pi_{#1}}
\def\MemorylessStrats#1{\Pi^{M}_{#1}}
\def\FairStrats#1{\Pi^{\mathcal{F}}_{#1}}
\def\MemorylessFairStrats#1{\Pi^{M\mathcal{F}}_{#1}}
\def\SemiMarkovFairStrats#1{\Pi^{S\mathcal{F}}_{#1}}
\def\SemiMarkovStrats#1{\Pi^{S}_{#1}}
\def\DetMemorylessStrats#1{\Pi^{MD}_{#1}}
\def\DetMemorylessFairStrats#1{\Pi^{MD\mathcal{F}}_{#1}}
\def\DetStrats#1{\Pi^{D}_{#1}}
\def\uniformstrat#1{\pi^{\mathrm{u}}_{#1}}

\def\U{\mathbin{\textsf{U}}}
\def\UpperBound{\mathbf{U}}
\def\LTL{\textsf{LTL}}
\def\PCTL{\textsf{PCTL}}
\def\Prism{\textsf{Prism}}

\newcommand{\postmin}{\post^{\min}}
\newcommand{\disttoT}[1]{{\lVert{#1}\rVert}}

\usepackage{color}

\definecolor{lightblue}{RGB}{231,255,255}
\definecolor{lightred}{RGB}{255,224,224}
\definecolor{lightgreen}{RGB}{224,255,224}
\definecolor{lightyellow}{RGB}{255,255,224}
\definecolor{lightpurple}{RGB}{255,224,255}
\definecolor{darkerred}{RGB}{64,0,0}
\definecolor{darkred}{RGB}{128,0,0}
\definecolor{darkblue}{RGB}{0,0,128}
\definecolor{darkgreen}{RGB}{0,128,0}
\definecolor{darkpurple}{RGB}{128,0,128}
\definecolor{black}{RGB}{0,0,0}

\makeatletter
\def\THICKhrulefill{\leavevmode \leaders \hrule height 5pt\hfill \kern \z@}
\makeatother

\mathchardef\mhyphen="2D

\urldef{\mailsa}\path|{pedro.dargenio,rdemasi}@unc.edu.ar|
\urldef{\mailsb}\path|{pcastro,lputruele}@dc.exa.unrc.edu.ar|

\graphicspath{{./Figs/}}

\title{Playing Against  Fair Adversaries in Stochastic Games with Total Rewards%
  \thanks{This work was supported by ANPCyT PICT-2017-3894 (RAFTSys),  ANPCyT PICT 2019-03134, SeCyT-UNC 33620180100354CB (ARES), and EU Grant agreement ID: 101008233 (MISSION).}\thanks{\textbf{A revised version of this paper was accepted in CAV 22.}}
}

\titlerunning{Playing Against  Fair Adversaries in Stochastic Games with Total Rewards} 

\author{Pablo F. Castro \inst{1,3} \and
Pedro R. D'Argenio \inst{2,3,4} \and \\
Luciano Putruele \inst{1,3}  \and 
Ramiro Demasi \inst{2,3}
}

\authorrunning{P.F. Castro et al.}
\institute{Departamento de Computaci\'on, FCEFQyN, Universidad Nacional de 
  R\'{\i}o Cuarto,
  Argentina
  \and 
  FAMAF, Universidad Nacional de C\'ordoba
  Argentina,
  \and
  Consejo Nacional de Investigaciones Cient\'ificas y T\'ecnicas (CONICET), Argentina \and
  Saarland University, Saarland Informatics Campus,
  Germany}

\begin{document}

\maketitle

\begin{abstract}
We investigate zero-sum turn-based two-player stochastic games in which the objective of one player is to maximize the amount of rewards obtained during a play, while
the other 
aims at minimizing it. 
We focus on games in which the minimizer plays in a 
fair way. We believe that these kinds of games enjoy interesting applications in software verification, where the maximizer plays the role of a system intending to maximize the
number of  ``milestones'' achieved, and the minimizer represents the behavior of some uncooperative but yet fair environment.
Normally, to study total reward properties, games are requested to be stopping (i.e., they reach a terminal state with probability 1).  
We relax the property to request that the game is stopping only under a fair minimizing player.
We prove that these games are determined, i.e., each state of the game has a value defined. Furthermore, we show that both
players have memoryless and deterministic optimal strategies, and the game value can be computed by approximating the greatest-fixed point of a set of functional equations. We implemented our approach in a prototype tool, and evaluated it on an illustrating example and an Unmanned Aerial Vehicle case study.  

\end{abstract}

\setcounter{page}{1}

\section{Introduction} \label{sec:intro}
	Game theory \cite{MorgensternNeuman42}  admits an elegant and profound mathematical theory. 
	In the last decades, it has received widespread attention from computer scientists because it has important applications to software synthesis and verification. 
	The analogy is appealing, the operation of a system under an uncooperative environment (faulty hardware, malicious agents, unreliable communication channels, etc.) can be modeled as a game 
	between two players (the system and the environment), in which the system tries to fulfill certain goals, whereas the environment tries to prevent this from happening. 
	This view is particularly useful for \emph{controller synthesis}, i.e., to automatically generate decision-making policies from high-level specifications. 
	Thus, synthesizing a controller consists of computing optimal strategies for a given game.
		
	In this paper we  focus on zero-sum, perfect information, two-player, turn-based stochastic games with (non-negative) rewards \cite{FilarV96}. 
	Intuitively, these games are played in a graph by two players who move a token in turns. Some vertices are probabilistic, in the sense that, if a token is in a probabilistic vertex, then 
	the next vertex is randomly selected. Furthermore, the players select their moves using strategies. Associated with each vertex there is a reward (which, in this paper, is taken to be non-negative).
	The goal of Player $1$ is to maximize the expected amount of collected rewards during the game, whereas Player $2$ aims at minimizing this value. 
	This is what~\cite{SvorenovaKwiatkowska16} calls \emph{total reward objective}.
	These kinds of games have been shown useful to reason about several classes of systems such as autonomous vehicles, 
	fault-tolerant systems, communication protocols, energy production plants, etc.  
        Particularly, in this paper we consider those games in which one of the players employs fair strategies. 

	Fairness restrictions, understood as fair resolutions of non-determinism of actions, play an important role in software verification and controller synthesis. 
	Especially, fairness assumptions over environments make possible the verification of liveness properties on open systems. 
	Several authors have indicated the need for fairness assumptions over the environment 
	in the controller synthesis approach, e.g., \cite{DBLP:conf/fossacs/AsarinCV10,DBLP:conf/icse/DIppolitoBPU11}.
        As a simple example consider an autonomous vehicle that needs to traverse a field where moving objects may interfere in its path.  Though the precise behaviour of the objects may be unknown, it is reasonable to assume that they will not continuously obstruct the vehicle attempts to avoid them.  In this sense, while stochastic behaviour may be a consequence of the vehicle faults, we can only assume a fair behaviour of the surrounding moving objects.
        In this work, we consider stochastic games in which one of the players (the one playing the environment) is assumed to play only with strong fair strategies.

        In order to guarantee that the expected value of accumulated rewards is well defined in (perhaps infinite) plays, some kind of stopping criteria is needed.  A common way to do this is to force the strategies to decide to stop with some positive probability in every decision.  This corresponds to the so-called discounted stochastic games~\cite{Shapley1095,FilarV96}, and has the implications that the collected rewards become less important  as the game progresses (the ``importance reduction'' is given by the discount factor).  Alternatively, one may be interested in knowing the expected \emph{total} reward, that is, the expected accumulated reward \emph{without} any loss of it as time progresses.  For this value to be well defined, the game itself needs to be stopping. That is, no matter the strategies played by the players, the probability of reaching a terminal state needs to be $1$~\cite{Condon90,FilarV96}.
        We focus on this last type of game.  However, we study here games that may not be stopping in general (i.e. for every strategy), but instead, require that they become stopping only when the minimizer plays in a fair way.  
        We use a  notion of (almost-sure) strong fairness, mostly following the ideas introduced in~\cite{DBLP:journals/dc/BaierK98} for Markov decision processes. We show that these kinds of games are determined, i.e., each state of the game has a value defined.  Furthermore, we show that memoryless and deterministic optimal  strategies exist for both players. Moreover, the value of the game can be calculated via the greatest fixed point of the corresponding functionals. 
It is important to remark that most of the properties discussed in this paper hold when the fairness assumptions are  made over the minimizer. Similar properties may not hold if the role of players is changed. 
However, these conditions encompass a large class of scenarios, where the system intends to maximize the total collected reward and the environment has the opposite objective.

In summary,  the contributions of this paper are the following: (1) we introduce  the notion of stopping under fairness stochastic game, a generalization of stopping game that takes into account fair environments; (2) we prove that it can be decided in polynomial time whether a game is stopping under fairness; (3) we show  that  these kinds of games are determined and both players possess optimal stationary strategies,  which can be computed using Bellman equations; and (4) we implemented these ideas in a prototype tool,  which was used  to evaluate practical case studies.

The paper is structured as follows. Section \ref{sec:mot_example} introduces an illustrating example to motivate the use of having fairness restrictions over the minimizer.  Section \ref{sec:background} fixes terminology and introduces background concepts. 
In Section \ref{sec:fair-strats} we describe a polynomial procedure to check whether a game stops under fairness assumptions, 
we also prove that determinacy is preserved in these games as well as the existence of (memoryless and deterministic) optimal strategies. 
Experimental results are described in Section \ref{sec:experimental_eval}. 
Finally, Sections \ref{sec:related_work} and  \ref{sec:conclusions} discuss related work and draw some conclusions, respectively.  Full proofs are gathered in the Appendix.


\newcommand{\roborta}{Roborta\xspace}

\section{\roborta vs.\ the Fair Light (A Motivating Example)} \label{sec:mot_example}

\begin{wrapfigure}[18]{l}{78mm}
\vspace{-11.6mm}
{\fontsize{6.6}{6.6}\selectfont\ttfamily
\begin{tabbing}
x\=xxxxxxxxxx\=xxxxxx\=xxx\=xxxx\=xxx\=xxxx\= \kill    
module \roborta\_vs\_the\_light\\[1ex]
\>col : [0..WIDTH] init 0; \\
\>row : [0..LENGTH] init 0; \\
\>light : [0..3] init 0; \>\>\>\> // current light color \\
\>                   \>\>\>\>// 0: red (light's turn) \\
\>                   \>\>\>\>// 1: yellow (\roborta moves sideways) \\
\>                   \>\>\>\>// 2: green (\roborta moves foreward) \\
\>                   \>\>\>\>// 3: off (light fails, any move) \\[1ex]
\> // light moves \\[1ex]
\>[l\_y] (light=0) \> \>-> \>(1-Q) : (light'=1) + Q : (light'=3);\\[1ex]
\>[l\_g] (light=0) \> \>-> \>(1-Q) : (light'=2) + Q : (light'=3);\\[1ex]
\> // \roborta moves \\[1ex]
\>[r\_l]  ((light=1) | (light=3)) \& (MOVES[col,row] <= 1)  \\
\>                    \>\>-> \>(1-P) : (light'=0) \& (col'=(col-1)\%WIDTH) + \\       
\>                     \>\>\>  P : (light'=0) ; \\[1ex]

\>[r\_r] ((light=1) | (light=3)) \& (MOVES[col,row] >= 1)\\
\>                    \>\>-> \> (1-P) : (light'=0) \& (col'=(col+1)\%WIDTH) + \\
\>                     \>\>\> P : (light'= 0); \\[1ex]
\>[r\_f] ((light=2) | (light=3)) \& (row < LENGTH) \\
\>                    \>\>-> \> (1-P) : (light'=0) \& (row'=row+1)  + \\
\>                     \>\>\> P : (light'= 0);\\[1ex]
endmodule\\[-5ex]
\end{tabbing}}
\caption{Model for the  Game} \label{fig:robot_game_model}
\end{wrapfigure}
		Consider the following scenario. \roborta the robot is navigating a grid of $4 \times 4$ cells. There is a light signal  
determining the possible movements of the robot: if the light is yellow, she must move  sideways (at a border cell, \roborta is allowed to wrap around to the other side); if the light is green she ought to move forward; if the light is red, she cannot perform any movement; if the light is off, the robot is free to move either sideways or forward.  The light signal and \roborta change their states in turns.  A (non-negative) reward is associated with each location of the grid.  Also, the sideway movements can be restricted to only one direction in some locations.
Moreover, we consider possible failures on the behavior of the robot and the light. If \roborta fails, she looses her turn to move, and if the light fails, it turns itself off.  The failures occur with a given probability. The goal of \roborta is to collect as much rewards as possible.  In opposition, the light aims at minimizing this value.
The specification of this game is captured in Fig.~\ref{fig:robot_game_model} (using {\Prism}-like notation \cite{DBLP:conf/cav/KwiatkowskaNP11}). 
In this model, \texttt{WIDTH} and \texttt{LENGTH} are constants defining the dimension of the grid.  \texttt{MOVES} is a two-dimensional array modeling the possible sideways movements in the grid (\texttt{0} allows the robot to move only to the left, \texttt{1}, to either side, and \texttt{2}, only to the right).  \texttt{P} and \texttt{Q}  define the failure probabilities of \roborta and the light respectively.
Fig.~\ref{fig:robot_game_grid} shows the assignment of rewards to each location of the $4 \times 4$ grid as well as the sideway movement restrictions (shown on the bottom-right of each location with white arrows).
The game starts at the location $(0, 0)$ and it stops when \roborta escapes through the end of the grid.

\begin{wrapfigure}[11]{r}{48mm}
\vspace{-9mm}
{\fontsize{6.6}{6.6}\selectfont\ttfamily
\centering
\includegraphics[scale=0.45]{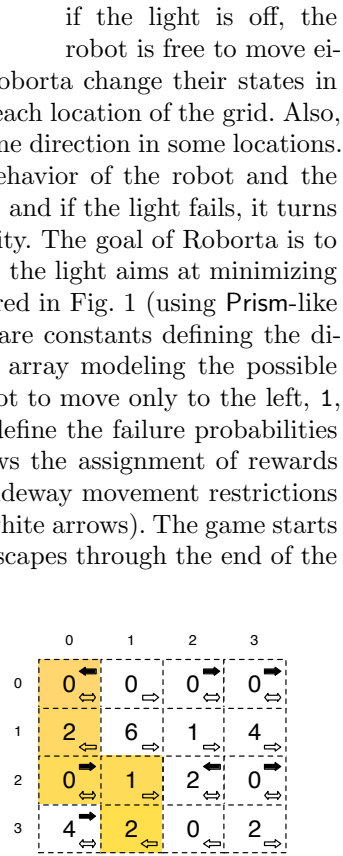}\hspace{1em}\mbox{}}
\caption{A robot on a $4 \times 4$ grid} \label{fig:robot_game_grid}
\end{wrapfigure}
	A possible scenario in this game is as follows. \roborta starts in cell $(0,0)$ and, in an attempt to minimize the rewards accumulated by the robot, the environment switches the yellow light on.
	For the sake of simplicity, we assume no failures on the light, i.e., $\texttt{Q}=0$. 
	Notice that, if the environment plays always in this way (signaling a yellow light), then \roborta will never achieve the goal and 
the game never stops.  This scenario occurs when the light plays in an unfair way, i.e., an action (the one that turns the green light on) is enabled infinitely often, but it is not executed infinitely often.
	Assuming fairness over the environment, we can ensure that a green light will be eventually switched on, allowing the robot to move forward.

	The best strategy for \roborta when the light is yellow is shown in black arrows on the top-right of the locations without movement restrictions.
	As a result, when both players play their optimal strategies, the path taken by \roborta to achieve the goal can be observed in the yellow-highlighted portion of the grid  in
Fig.~\ref{fig:robot_game_grid}.  In Section \ref{sec:experimental_eval}, we evaluate this problem experimentally with different 
configurations of this game.

\section{Preliminaries} \label{sec:background}

We  introduce some basic definitions and results on stochastic games that will be necessary across the paper.

A (discrete) \emph{probability distribution} $\mu$ over a denumerable 
set $S$ is a function $\mu: S \rightarrow [0, 1] $  such that 
$\mu(S) = \sum_{s \in S} \mu(s) = 1$. 
Let $\Dist(S)$ denote the set of all probability distributions on $S$. $\Delta_s \in \Dist(S)$ denotes the Dirac distribution for $s \in S$, i.e., 
$\Delta_s(s) =1$ and $\Delta_s(s') = 0$ for all $s' \in S$ such that $s'\neq s$.
The \textit{support} set of $\mu$ is defined by $\support{\mu} = \{s |~\mu (s) > 0\}$.

	 Given a set $V$, $V^*$ (resp. $V^\infty$) denotes the set of all finite sequences (resp. infinite sequences) of elements of $V$. Concatenation is represented using juxtaposition.  We use variables $\omega, \omega', \dots \in V^\infty$ as ranging over infinite sequences, and variables $\hat{\omega}, \hat{\omega}', \dots \in V^*$ as ranging over finite sequences. The $i$-th element of a finite (resp. infinite) sequence $\hat{\omega}$ (resp. $\omega$) is denoted 
$\hat{\omega}_i$ (resp. $\omega_i$). Furthermore, for any finite sequence $\hat{\omega}$, $|\hat{\omega}|$ denotes its length. For $\omega \in V^\infty$, $\inf(\omega)$ denotes the set of items appearing infinitely often in $\omega$. Given
$S \subseteq V^*$, $S^k$ is the set obtained by concatenating $k$ times the sequences in $S$.
	 
A \emph{stochastic game} \cite{ChatterjeeH12,SvorenovaKwiatkowska16} is a tuple $\StochG = ( V,  (V_1, V_2, V_\Probabilistic), \delta  ) $, where $V$ is a finite set of vertices (or states) with $V_1, V_2, V_\Probabilistic \subseteq V$ being a partition of $V$, and $\delta : V \times V \rightarrow [0,1]$ is a probabilistic transition function, such that for every $v \in V_1\cup V_2$, $\delta(v,v') \in  \{0,1\}$, for any $v' \in V$; and $\delta(v,\cdot) \in \Dist(V)$ for $v \in V_\Probabilistic$.
If $V_\Probabilistic = \emptyset$, then $\mathcal{G}$ is called a two-player game graph. Moreover, if $V_1 = \emptyset$ or $V_2 = \emptyset$, then $\mathcal{G}$ is a \emph{Markov decision process} (or MDP). Finally, in case that $V_1= \emptyset$ and $V_2 = \emptyset$, $\mathcal{G}$ is a \emph{Markov chain} (or MC). For all states $v \in V$ we define $\post^\delta(v) = \{v' \in V \mid \delta(v,v')>0 \}$, the set of successors of $v$. Similarly, $\pre^\delta(v') = \{v \in V \mid \delta(v,v')>0 \}$ as the set of predecessors of $v'$, we omit the index $\delta$ when it is clear from context.  Also, when useful, we fix an initial state for a game, in such a case we use the notation $\mathcal{G}_v$ to indicate that the game starts from $v$. Furthermore, we assume that $\post(v) \neq \emptyset$ for every
$v \in V$.
A vertex $v \in V$ is said to be \emph{terminal} if $\delta(v,v) = 1$, and $\delta(v,v')=0$ for all $v \neq v'$.
Most results on MDPs rely on the notion of \emph{end component} \cite{BaierK08}, we straightforwardly extend this notion to two-player games: an end component of $\StochG$ is a pair $(V',\delta')$ such that (a) $V' \subseteq V$; (b) $\delta'(v) = \delta(v)$ for $v \in V_\Probabilistic$; (c) $\emptyset \neq \post^{\delta'}(v) \subseteq \post^\delta(v)$ for $v \in V_1 \cup V_2$; (d) $\post^{\delta'}(v) \subseteq V'$ for all $v \in V'$; (e) the underlying graph of $(V',\delta')$ is strongly connected.  Note that an end component can also be considered as being a game. The set of end components of $\StochG$ is denoted $\EndComp(\StochG)$.
 
%
%
A \emph{path} in the game $\StochG$ is an infinite sequence of vertices $v_0 v_1 \dots$ such that $\delta(v_k, v_{k+1})>0$ for every $k \in \mathbb{N}$.
$\GamePaths_{\StochG}$ denotes the set of all paths,  and $\FinGamePaths_{\StochG}$ denotes the set of finite prefixes of paths. 
Similarly, $\GamePaths_{\StochG,v}$  and  $\FinGamePaths_{\StochG,v}$ denote the set of paths and the set of finite paths starting at vertex $v$.
	 

A \emph{strategy} for Player $i$ (for $i\in\{1,2\}$) in a game $\StochG$ is a function $\strat{i}: V^*  V_i \rightarrow \Dist(V)$ that assigns a probabilistic distribution to each finite sequence of states such that $\strat{i}(\hat{\omega}v)(v') > 0$ only if $v' \in \post(v)$.  The set of all the strategies for Player $i$ is named $\Strategies{i}$. A strategy $\strat{i}$ is said to be  \emph{pure} or \emph{deterministic} if, for every $\hat{\omega} v \in V^*V_i$, $\strat{i}(\hat{\omega} v)$  is a Dirac distribution, and it is called \emph{memoryless} if $\strat{i}(\hat{\omega} v) = \strat{i}(v)$, for every $\hat{\omega} \in V^*$.
%
%
Let $\MemorylessStrats{i}$ and $\DetStrats{i}$ be respectively
the set of all memoryless strategies and 
the set of all deterministic strategies for Player $i$.  $\DetMemorylessStrats{i} = \MemorylessStrats{i} \cap \DetStrats{i}$ is the set of all its deterministic and memoryless strategies.

	Given two strategies $\strat{1} \in \Strategies{1}$, $\strat{2} \in \Strategies{2}$ and an initial vertex $v$,  the \emph{result} of the game  is a Markov chain \cite{ChatterjeeH12}, denoted 
$\StochG^{\strat{1}, \strat{2}}_v$. An event $\mathcal{A}$ is a measurable set in the Borel $\sigma$-algebra generated by the cones of $\GamePaths_{\StochG}$.  The \emph{cone} or \emph{cylinder} spanned by the finite path $\hat{\omega} \in \FinGamePaths_{\StochG}$ is the set $\cone(\hat{\omega})=\{\omega \in \GamePaths_\StochG \mid \forall 0 \leq i < |\hat{\omega}|: \omega_i = \hat{\omega}_i \}$. $\Prob{\strat{1}}{\strat{2}}_{\StochG,v}$ is the associated probability measure obtained when fixing strategies $\strat{1}$, $\strat{2}$, and an initial vertex $v$  \cite{ChatterjeeH12}. Intuitively, $\Prob{\strat{1}}{\strat{2}}_{\StochG,v}(\mathcal{A})$
is the probability that strategies $\strat{1}$ and $\strat{2}$ generates a path belonging to the set $\mathcal{A}$ when the game $\StochG$ starts in $v$. When no confusion is possible, we just write $\Prob{\strat{1}}{\strat{2}}_{\StochG,v}(\hat{\omega})$ instead of $\Prob{\strat{1}}{\strat{2}}_{\StochG,v}(\cone(\hat{\omega}))$.
Similar notations are used  for MDPs and MCs.  A stochastic game (defined as above) is said to be \emph{stopping} \cite{Condon92} if 
for all pair of strategies $\strat{1}, \strat{2}$ the probability of reaching a terminal state is $1$. We  use {\LTL} notation to represent specific set of paths, e.g., $\Diamond T = \{\omega \in \GamePaths_{\StochG} \mid \exists i \geq 0 : \omega_i \in T \}$ is the set of all the plays in the game that reach vertices in $T$.




A \emph{quantitative objective} or \emph{payoff function} is a measurable  function $f: V^{\infty} \to \mathbb{R}$.   Let $\Expect{\strat{1}}{\strat{2}}_{\StochG,v}[f]$ be the expectation of measurable
function $f$ under probability $\Prob{\strat{1}}{\strat{2}}_{\StochG,v}$.  The goal of Player 1 is to maximize this value whereas the goal of 
Player $2$ is to minimize it.  Sometimes quantitative objective functions can be defined via \emph{rewards}. These are assigned by a \emph{reward function} $\reward:V \to \mathbb{R}^+$.  A \emph{stochastic game with rewards} is a tuple $(V, (V_1, V_2, V_\Probabilistic), \delta, \reward \rangle$ composed of a stochastic game and a reward function.  We assume that for every terminal vertex $v$,  $\reward(v) = 0$.

	The value of the game for Player $1$ at vertex $v$ under strategy $\strat{1}$ is defined as the infimum over all the values  
resulting from Player $2$ strategies in that vertex, i.e. $\inf_{\strat{2} \in \Strategies{2}} \Expect{\strat{1}}{\strat{2}}_{\StochG,v}[f]$.
The \emph{value of the game} for Player $1$ is defined as the supremum of the values of all Player $1$ strategies, i.e., $\sup_{\strat{1} \in \Strategies{1}} \inf_{\strat{2} \in \Strategies{2}} \Expect{\strat{1}}{\strat{2}}_{\StochG,v}[f]$.
Analogously, the value of the game for a Player $2$ under strategy $\strat{2}$ and the value of the game 
for Player $2$ are defined as $\sup_{\strat{1} \in \Strategies{1}}  \Expect{\strat{1}}{\strat{2}}_{\StochG,v}[f]$ 
and $\inf_{\strat{2} \in \Strategies{2}} \sup_{\strat{1} \in \Strategies{1}}  \Expect{\strat{1}}{\strat{2}}_{\StochG,v}[f]$, respectively. We say that a game is \emph{determined} if both values are the same, that is,
	$\sup_{\strat{1} \in \Strategies{1}} \inf_{\strat{2} \in \Strategies{2}} \Expect{\strat{1}}{\strat{2}}_{\StochG,v}[f]
	=
	\inf_{\strat{2} \in \Strategies{2}} \sup_{\strat{1} \in \Strategies{1}} \Expect{\strat{1}}{\strat{2}}_{\StochG,v}[f]$.
	Martin \cite{Martin98} proved the determinacy of stochastic games for Borel and bounded objective functions.
	
	In this paper we focus  on the  \emph{total accumulated reward payoff} function, i.e.: $\Rewards(\omega) = \sum^\infty_{i=0} \reward(\omega_i)$. Since $\Rewards$ is unbounded, the results of Martin \cite{Martin98} do not apply to this function.  In this paper we restrict ourselves to non-negative rewards, as shown in the next sections, non-negative rewards are enough to deal with interesting case studies, we briefly discuss in Section \ref{sec:conclusions} the possible extension of the results presented here to games having negative rewards.
\section{Stopping Games and Fair Strategies}\label{sec:fair-strats}

We begin this section by introducing the notions of \emph{(almost sure) fair strategy} and \emph{stopping games under fairness}. 
From now on, we assume that Player $2$ represents the environment,  which tries to minimize the amount of rewards obtained by the system, thus  fairness restrictions will be applied to this player. 

\begin{definition}
   Given a stochastic game $\StochG = (V, (V_1, V_2, V_\Probabilistic), \delta)$.
The set of fair plays for Player $2$ (denoted $\FP^2$) is defined as follows:
\[
	\FP^2 = \{ \omega \in \GamePaths_{\StochG} \mid \forall v' \in V_2: v' \in \inf(\omega)  \Rightarrow \post(v') \subseteq \inf(\omega) \}
\]
\end{definition}
	Alternatively,  if we consider each vertex as a proposition,  $\FP^2$ can be written using {\LTL} notation as: 
	$\bigwedge_{v \in V_2} \bigwedge_{v' \in \post(v)}(\Box \Diamond v \Rightarrow \Box \Diamond v')$.  This property is 	 $\omega$-regular, thus it is measurable  in the $\sigma$-algebra generated by the cones of $\GamePaths_{\StochG}$ (see e.g. \cite[p.804]{BaierK08}.)  This is a state-based notion of fairness, but it can be straightforwardly extended to settings where transitions are considered. For the sake of simplicity we do not do so in this paper.
	
	
	Next, we introduce the notion of  (almost-sure) \emph{fair strategies} for Player $2$.
\begin{definition} Given a stochastic game $\StochG = (V, (V_1, V_2, V_\Probabilistic), \delta)$,
a strategy $\strat{2} \in \Strategies{2}$ is said to be \emph{almost-sure fair} (or simply \emph{fair}) iff it holds that:
$\Prob{\strat{1}}{\strat{2}}_{\StochG,v}(\FP^2) = 1$,
for every $\strat{1} \in \Strategies{1}$ and $v \in V$. 
\end{definition}
The set of all the  fair strategies for Player $2$ is denoted by $\FairStrats{2}$. We combine this notation  with the notation introduced in Section \ref{sec:background}, e.g.: $\MemorylessFairStrats{2}$ refers to the set of all  memoryless and fair strategies for Player $2$. 
The previous definition is based on the notion of fair scheduler as introduced for Markov decision processes~\cite{DBLP:journals/dc/BaierK98,BaierK08}.

Note that for stopping games, every strategy is fair, because the probability of visiting a vertex infinitely often is $0$.
Also notice that there are games which are not stopping, but they become stopping if Player $2$ uses only fair strategies. This is the main idea behind the notion of \emph{stopping under fairness} as introduced in the following definition.
	
\begin{definition}\label{def:stopping-under-fairness}
  A stochastic game $\StochG = (V, (V_1, V_2, V_\Probabilistic), \delta)$ is said  to be \emph{stopping under fairness} if for all strategies $\strat{1} \in \Strategies{1}, \strat{2} \in \FairStrats{2}$ and vertex $v \in V$, it holds that
  $\Prob{\strat{1}}{\strat{2}}_{\StochG,v}(\Diamond T)=1$,  where $T$ is the set of terminal vertices of $\StochG$. 
\end{definition}

\paragraph{Checking the stopping criteria.}

This section is devoted to the effective
characterization of games that are stopping under fairness.
The
following lemma states that for every game that is not stopping under
fairness there is a \emph{memoryless deterministic} strategy for
Player 1 and a fair strategy for Player 2 that witnesses it.

\begin{lemma}\label{lemma:memoryless-strat}
  Let $\StochG = (V, (V_1, V_2, V_\Probabilistic), \delta)$ be a
  stochastic game, $v \in V$, and $T$ the set of terminal states of $\StochG$.
  If $\Prob{\strat{1}}{\strat{2}}_{\mathcal{G}, v}(\Diamond T) < 1$
  for some
  $\strat{1} \in \Strategies{1}$ and $\strat{2} \in \FairStrats{2}$,
  then, for some memoryless and deterministic strategy
  $\strat{1}' \in \DetMemorylessStrats{1}$ and fair strategy
  $\strat{2}' \in \FairStrats{2}$,
  $\Prob{\strat{1}'}{\strat{2}'}_{\mathcal{G}, v}(\Diamond T) < 1$.
\end{lemma}

The proof of this lemma follows by noticing that, if
$\Prob{\strat{1}}{\strat{2}}_{\mathcal{G}, v}(\Diamond T) < 1$, there
must be a finite path that leads with some probability to an end
component not containing a terminal state and which is a trap for the
fair strategy $\strat{2}$.  This part of the game enables the
construction of a memoryless deterministic strategy for Player 1 by
ensuring that it follows the same finite path (but skipping loops) and
that it traps Player 2 in the same end component.

The next theorem states that checking stopping under fairness in a
stochastic game $\StochG$ can be reduced to check the stopping
criteria in the MDP obtained from $\StochG$ by fixing a strategy in
Player 2 that selects among the output transitions according to a
uniform distribution.  Thus, this theorem enables a graph solution to
determine stopping under fairness.


\begin{theorem}\label{thm:uniform-prob}
  Let $\StochG = (V, (V_1, V_2, V_\Probabilistic), \delta)$ be a
  stochastic game and $T$ its set of terminal states.
  Consider the Player 2 (memoryless) strategy
  $\uniformstrat{2} : V_2 \rightarrow \Dist(V)$ defined by
  $\uniformstrat{2}(v)(v') = \frac{1}{\# \post(v)}$, for all $v \in
  V_2$ and $v' \in post(v)$.
  Then, $\StochG$ is stopping under fairnes iff
  $\Prob{\strat{1}}{\uniformstrat{2}}_{\StochG,v}(\Diamond T)=1$ for
  every $v\in V$ and $\strat{1} \in \Strategies{1}$.
\end{theorem}

While the ``only if'' part of the theorem is direct, the ``if'' part
is proved by contraposition using Lemma~\ref{lemma:memoryless-strat}.

In practice, one needs some procedure to check whether a stochastic
game is stopping under fairness. As already remarked above, a game
could be not stopping for all pairs of strategies, but it could be
stopping when fairness assumptions are made about Player $2$'s
behavior.
One could use Theorem~\ref{thm:uniform-prob} to transform game
$\StochG$ into the MDP $\StochG^{\uniformstrat{2}}$ by fixing
$\uniformstrat{2}$ in $\StochG$ and check whether
$\MDPProb{\strat{1}}_{\StochG^{\uniformstrat{2}},v}(\Diamond T)=1$ for
all $v\in V$.  Instead, we use Theorem~\ref{thm:uniform-prob} to
provide a direct algorithm. Using this, we prove that we can decide in
polynomial time whether a game is stopping under fairness or not.
The main idea is to use a modification of the standard $\pre$
operator, as shown in the following definition:
\begin{align*}
  \EFairpre(C) \ = \ {}&\{ v \in V \mid \delta(v,C) > 0\} \\
  \AFairpre(C) \ = \ {}&\{ v \in V_2\cup V_\Probabilistic \mid \delta(v,C)>0\} \\
                       &\cup \ \{ v \in  V_1 \mid \forall v' \in V : \delta(v,v') > 0 \Rightarrow v' \in C \}
\end{align*}
As usual we consider the transitive closures of these operators denoted $\EFairpre^*$ and $\AFairpre^*$, respectively.
%
\begin{theorem}\label{thm:stopping-algorithm}
  Let $\StochG = (V, (V_1, V_2, V_\Probabilistic), \delta)$, be a
  stochastic game and let $T$ be the set of its terminal states. Then,
  \begin{inparaenum}[(1)]
  \item%
    $\Prob{\strat{1}}{\strat{2}}_{\StochG,v}(\Diamond T) = 1$ for every
    $\strat{1} \in \Strategies{1}$ and $\strat{2} \in \FairStrats{2}$
    iff $v \in V\setminus \EFairpre^*(V \setminus \AFairpre^*(T))$, and
  \item%
    $\StochG$ is stopping under fairness iff
    $\EFairpre^*(V \setminus \AFairpre^*(T)) = \emptyset$.
  \end{inparaenum}
\end{theorem}

It is direct to see that $\EFairpre$ and $\AFairpre$ can be computed using standard graph algorithms which are polynomial in time w.r.t.\ the size of the game, thus we obtain the following theorem.

\begin{theorem}\label{th:fair-is-poly}
  Checking whether the stochastic game $\StochG$ is stopping under
  fairness or not is in $O(\poly(\size(\StochG)))$.
\end{theorem}

	
\paragraph{Determinacy of Stopping Games under Fairness.}

The determinacy of stochastic games with Borel and bounded payoff functions follows from Martin's results~\cite{Martin98}.  The  function $\Rewards$ is unbounded, so Martin's theorems do not apply to it. In \cite{FilarV96}, the determinacy of a general class of stopping stochastic games (called \emph{transient}) with total rewards is proven.  However, note that we restrict Player $2$ to only play with fair strategies and hence, the last result does not apply either.
In \cite{PatekBertsekas99} the authors classify Player $2$'s strategies into proper (those ensuring termination) and improper (those prolonging the game indefinitely). For proving determinacy, the authors assume that the value of the game for Player $2$'s  improper strategies is $\infty$.  It is worth noting that, for proving the results below, we do not make any assumption about unfair strategies.
%
In the following we prove that the restriction to fair plays does not affect the determinacy of the  games.

The notion of \emph{semi-Markov} strategies \cite{FilarV96} will be important for the theorems in this section. Intuitively, a semi-Markov strategy only takes into account the length of a play, the initial state, and the current state to select the next step in the play.  
        
\begin{definition} Let $\StochG = (V, (V_1, V_2, V_\Probabilistic), \delta)$ be a stochastic game. A strategy $\strat{i} \in \Strategies{i}$ is called semi-Markov if: $\strat{i}(v \hat{\omega}v') = \strat{i}(v \hat{\omega}'v')$, for every $v \in V$ and $\hat{\omega}, \hat{\omega}' \in V^*$ such that $|\hat{\omega}|=|\hat{\omega}'|$. 
\end{definition}

The set of all semi-Markov (resp. semi-Markov fair) strategies for player $i$ is denoted $\SemiMarkovStrats{i}$ (resp. $\SemiMarkovFairStrats{i}$). The importance of semi-Markov strategies lies in the fact that, when Player $2$ plays a semi-Markov strategy,   any Player $1$'s strategy can be mimicked by a semi-Markov strategy as stated in the following theorem.

\begin{theorem}\label{th:semmimarkov2}
  Let $\StochG$ be a stopping under fairness stochastic game, and let
  $\strat{2} \in \SemiMarkovFairStrats{2}$ be a fair and semi-Markov
  strategy. Then, for any $\strat{1} \in \Strategies{1}$, there is a
  semi-Markov strategy $\starredstrat{1} \in \SemiMarkovStrats{1}$
  such that
  $\Expect{\strat{1}}{\strat{2}}_{\StochG, v}[\Rewards] =
  \Expect{\starredstrat{1}}{\strat{2}}_{\StochG, v}[\Rewards]$.
\end{theorem}
\begin{proof}[Sketch]
  The proof follows the arguments of Theorem 4.2.7 in \cite{FilarV96}
  adapted to our setting.
	
  Consider the event $\Diamond^{k} v' = \{ \omega \in
  \GamePaths_\StochG \mid \omega_k = v'\}$, for $k\geq 0$. That is,
  the set of runs in which $v'$ is reached after exactly $k$ steps.
  We define $\starredstrat{1}$ as follows.
  \[
  \starredstrat{1}(\hat{\omega}v')(v'') =  \Prob{\strat{1}}{\strat{2}}_{\StochG,v}(\Diamond^{k+1} v'' \mid \Diamond^k v') 
  \]
  for every $\hat{\omega}$ such that $\Prob{\strat{1}}{\strat{2}}_{\StochG,v}(\hat{\omega}v') > 0$ and $|\hat{\omega}v'| = k$.  For those $\hat{\omega}$'s with $\Prob{\strat{1}}{\strat{2}}_{\StochG,v}(\hat{\omega}v') = 0$ we define $\starredstrat{1}(\hat{\omega}v')$ to be any arbitrary distribution.   We prove that this strategy satisfies the conditions of the theorem.
  For this, we prove that $\Prob{\strat{1}}{\strat{2}}_{\StochG,v}(\Diamond^{k} v') = \Prob{\starredstrat{1}}{\strat{2}}_{\StochG,v}(\Diamond^{k} v')$ by induction on $k$.  
  %
  %
  %
  Hence:
  \begin{align*}
  \Expect{\strat{1}}{\strat{2}}_{\StochG,v}[\Rewards]   &  = \sum^\infty_{N=0} \sum_{\hat{\omega} \in V^{N+1}} \Prob{\strat{1}}{\strat{2}}_{\StochG,v}(\hat{\omega})\reward(\hat{\omega}_N)\\
  & = \sum^\infty_{N=0} \sum_{v \in V} \Prob{\strat{1}}{\strat{2}}_{\StochG,v}(\Diamond^N v)\reward(v) \\
  & =  \sum^\infty_{N=0} \sum_{v \in V} \Prob{\starredstrat{1}}{\strat{2}}_{\StochG,v}(\Diamond^N v)\reward(v) \\
  & = \Expect{\starredstrat{1}}{\strat{2}}_{\StochG,v}[\Rewards] 
  \end{align*}
\qed
\end{proof}
%

The results below relate the games presented in this paper with
transient games \cite{FilarV96,Altman99}.  Roughly speaking, the
following property states that, in games that stop under fairness, the
expected time that both players spend in non-terminal states is
finite, given that the two players play memoryless strategies.

\begin{lemma}\label{lemma:bound-prob-stationary-strats}
  Let $\StochG = (V, (V_1, V_2, V_\Probabilistic), \delta)$ be a
  stochastic game that is stopping under fairness with $T$ being the
  set of terminal states.
  Let $\strat{1} \in \MemorylessStrats{1}$ be a memoryless strategy
  for Player~1 and $\strat{2} \in \MemorylessFairStrats{2}$ a
  memoryless fair strategy for Player~2.  Then
  $\sum^\infty_{N=1} \sum_{\hat{\omega} \in (V \setminus T)^{N}} \Prob{\strat{1}}{\strat{2}}_{\StochG,v}(\hat{\omega}) < \infty$.
\end{lemma}

This result can be extended to all the strategies of Player~1. The
main idea behind the proof is to fix a stationary fair strategy for
Player~2 (e.g., a uniform distributed strategy).  This yields an MDP
that stops for every strategy of Player~1, and furthermore, it can be
seen as a one-player \emph{transient game} (as defined in
\cite{FilarV96}).  Hence, the result follows from Lemma
\ref{lemma:bound-prob-stationary-strats} and Theorem 4.2.12 in
\cite{FilarV96}.

\begin{theorem}\label{th:games-are-bounded}
  Let $\StochG$ be a stochastic game that is stopping under fairness
  and let $T$ be the set of terminal states.  In addition, let
  $\strat{1} \in \Strategies{1}$ be a strategy for Player $1$ and
  $\strat{2} \in \MemorylessFairStrats{2}$ be a fair and memoryless
  strategy for Player $2$.  Then
  $\sum^\infty_{N=0} \sum_{\hat{\omega} \in v(V \setminus T)^N} \Prob{\strat{1}}{\strat{2}}_{\StochG,v}(\hat{\omega}) < \infty$.
\end{theorem}

Using the previous theorem, some fairly simple calculations lead to
the fact that the value of the total accumulated reward payoff game is
well-defined for any strategy of the players.  This is stated in the
next theorem.

\begin{theorem}\label{th:memoryless-strat-p2-bounded-expectation}
  Let $\StochG= (V, (V_1, V_2, V_\Probabilistic), \delta)$ be a
  stochastic game that is stopping under fairness,
  $\strat{1} \in \Strategies{1}$ a strategy for Player~1, and
  $\strat{2} \in \FairStrats{2}$ a fair strategy for Player~2.
  Then, for all $v \in V$,
  $\Expect{\strat{1}}{\strat{2}}_{\StochG,v}[\Rewards] < \infty$.
\end{theorem}

As a consequence of this theorem, the value of the game is bounded
from above for any Player $1$'s strategy.
\begin{corollary}\label{coro:inf-for-strat2-is-bounded}
  Let $\StochG= (V, (V_1, V_2, V_\Probabilistic), \delta)$ be a
  stochastic game that is stopping under fairness and
  $\strat{1} \in \Strategies{1}$ a strategy for Player $1$.
  Then, for every vertex $v \in V$,
  $\inf_{\strat{2} \in \FairStrats{2}} \Expect{\strat{1}}{\strat{2}}_{\StochG,v}[\Rewards] < \infty$.
\end{corollary}

\begin{proof}
  Consider any fair strategy $\strat{2}' \in \FairStrats{2}$.
  Then,
  $\inf_{\strat{2} \in \FairStrats{2}} \Expect{\strat{1}}{\strat{2}}_{\StochG,v}[\Rewards]
  \leq \Expect{\strat{1}}{\uniformstrat{2}}_{\StochG,v}[\Rewards] < \infty$,
  where the last inequality follows from
  Theorem~\ref{th:memoryless-strat-p2-bounded-expectation}.
  \qed
\end{proof}

The following theorem plays an important role in the rest of the
paper.  Intuitively, it states that, when Player~1 plays with a
memoryless strategy, Player~2 has an optimal deterministic
memoryless fair strategy.

\begin{theorem}\label{th:infima-in-dmf}%
  Let $\StochG= (V, (V_1, V_2, V_\Probabilistic), \delta)$ be a
  stochastic game that is stopping under fairness and let
  $\strat{1} \in \MemorylessStrats{1}$ be a memoryless strategy for
  Player~1.  There exists a deterministic memoryless fair strategy
  $\starredstrat{2}\in\MemorylessFairStrats{2}$ such that
  $\inf_{\strat{2} \in \FairStrats{2}} \Expect{\strat{1}}{\strat{2}}_{\StochG,v}[\Rewards]
   =
   \Expect{\strat{1}}{\starredstrat{2}}_{\StochG,v}[\Rewards]$, for every $v \in V$.
\end{theorem}

As an immediate consequence of Theorem~\ref{th:infima-in-dmf}, we have
that, if Player 1 plays with a memoryless strategy, it suffices that
Player~2 considers only memoryless fair strategies to find its
optimal value:

\begin{corollary}\label{th:memoryless-infima}
  Let $\StochG= (V, (V_1, V_2, V_\Probabilistic), \delta)$ be a stochastic game that is stopping under fairness
  and let $\strat{1} \in \MemorylessStrats{1}$ be a memoryless
  strategy for Player~1. Then,
  $\inf_{\strat{2} \in \FairStrats{2}} \Expect{\strat{1}}{\strat{2}}_{\StochG,v}[\Rewards]
  =
  \inf_{\strat{2} \in \MemorylessFairStrats{2}} \Expect{\strat{1}}{\strat{2}}_{\StochG,v}[\Rewards]$, for every $v \in V$.
\end{corollary}

Semi-Markov strategies can be thought of as sequences of memoryless strategies. The next theorem uses this fact to show that, when Player $2$ plays a memoryless and fair strategy,  semi-Markov strategies do not improve the value that Player $1$ can obtain via memoryless strategies. The proof of the following theorem adapts the ideas of Theorem 4.2.9 in \cite{FilarV96} to our games.

\begin{theorem}\label{th:semimarkov-to-memoryless} For any stochastic game $\StochG$ that is stopping under fairness, and vertex $v$, it holds that:
\[\adjustlimits
	\sup_{\strat{1} \in \SemiMarkovStrats{1}} \inf_{\strat{2} \in \MemorylessFairStrats{2}} \Expect{\strat{1}}{\strat{2}}_{\StochG,v}[\Rewards]
	= \adjustlimits
	\sup_{\strat{1} \in \MemorylessStrats{1}} \inf_{\strat{2} \in \MemorylessFairStrats{2}} \Expect{\strat{1}}{\strat{2}}_{\StochG,v}[\Rewards]
\]
\end{theorem}

Using the previous theorem, we can prove the following
useful property.
\begin{theorem}\label{th:reduce-to-memoryless}
  For any stochastic game $\StochG$ that is stopping under fairness we have:
  \[
  \adjustlimits \sup_{\strat{1} \in \Strategies{1}} \inf_{\strat{2} \in \FairStrats{2}} \Expect{\strat{1}}{\strat{2}}_{\StochG,v}[\Rewards]
  =
  \adjustlimits \sup_{\strat{1} \in \MemorylessStrats{1}} \inf_{\strat{2} \in \MemorylessFairStrats{2}} \Expect{\strat{1}}{\strat{2}}_{\StochG,v}[\Rewards]
 \]
\end{theorem}
\begin{proof}
	First, we prove that the left-hand term is smaller or equal to the right-hand one:
\begin{align*}
  \adjustlimits \sup_{\strat{1} \in \Strategies{1}} \inf_{\strat{2} \in \FairStrats{2}} \Expect{\strat{1}}{\strat{2}}_{\StochG,v}[\Rewards]
  &{} \leq  \adjustlimits \sup_{\strat{1} \in \Strategies{1}} \inf_{\strat{2} \in \MemorylessFairStrats{2}} \Expect{\strat{1}}{\strat{2}}_{\StochG,v}[\Rewards]\\
  &{} \leq \adjustlimits  \sup_{\strat{1} \in \SemiMarkovStrats{1}} \inf_{\strat{2} \in \MemorylessFairStrats{2}} \Expect{\strat{1}}{\strat{2}}_{\StochG,v}[\Rewards]
  \leq \adjustlimits \sup_{\strat{1} \in \MemorylessStrats{1}} \inf_{\strat{2} \in \MemorylessFairStrats{2}} \Expect{\strat{1}}{\strat{2}}_{\StochG,v}[\Rewards]
\end{align*}
	The first inequality follows from $\MemorylessFairStrats{2} \subseteq \FairStrats{2}$, the second inequality is due to Theorem \ref{th:semmimarkov2} and the fact that memoryless strategies are semi-Markov, and the last inequality is obtained by applying Theorem \ref{th:semimarkov-to-memoryless}.
	
To prove the other inequality, we calculate:
\[\adjustlimits \sup_{\strat{1} \in \MemorylessStrats{1}} \inf_{\strat{2} \in \MemorylessFairStrats{2}} \Expect{\strat{1}}{\strat{2}}_{\StochG,v}[\Rewards]
=  \adjustlimits \sup_{\strat{1} \in \MemorylessStrats{1}} \inf_{\strat{2} \in \FairStrats{2}} \Expect{\strat{1}}{\strat{2}}_{\StochG,v}[\Rewards]
\leq  \adjustlimits \sup_{\strat{1} \in \Strategies{1}} \inf_{\strat{2} \in \FairStrats{2}} \Expect{\strat{1}}{\strat{2}}_{\StochG,v}[\Rewards].
\]
	The first equality follows from Corollary \ref{th:memoryless-infima}, and the second inequality is due to suprema properties.
\qed
\end{proof}
	Furthermore, since the values obtained via randomized memoryless strategies are convex combinations of  values
obtained via deterministic memoryless strategies, we obtain the following corollary.
\begin{corollary}\label{cor:value- game-deterministic-strats}For any stochastic game $\StochG$ that is stopping under fairness we have:
\[
	\adjustlimits \sup_{\strat{1} \in \Strategies{1}} \inf_{\strat{2} \in \FairStrats{2}} \Expect{\strat{1}}{\strat{2}}_{\StochG,v}[\Rewards]
	= \adjustlimits
	\sup_{\strat{1} \in \DetMemorylessStrats{1}} \inf_{\strat{2} \in \DetMemorylessFairStrats{2}} \Expect{\strat{1}}{\strat{2}}_{\StochG,v}[\Rewards]
\]
\end{corollary}	

	Summing up, we can compute the value $\sup_{\strat{1} \in \Strategies{1}} \inf_{\strat{2} \in \FairStrats{2}} \Expect{\strat{1}}{\strat{2}}[\Rewards]$, for any vertex $v$, by restricting ourselves to deterministic memoryless strategies. 
	
It is worth noting  that,  by Theorem~\ref{th:memoryless-strat-p2-bounded-expectation}, the value $\Expect{\strat{1}}{\strat{2}}_{\StochG,v}[\Rewards]$ is finite 
for every stopping game under fairness $\StochG$ and strategies $\strat{1} \in \DetMemorylessStrats{1}$, $\strat{2} \in \DetMemorylessFairStrats{2}$. 
 Furthermore, because the number of deterministic memoryless strategies is finite  we have that the number 
 $\UpperBound = \max \{ \inf_{\strat{2} \in \DetMemorylessFairStrats{2}}  \sup_{\strat{1} \in \DetMemorylessStrats{1}} \Expect{\strat{1}}{\strat{2}}_{\StochG,v}[\Rewards] \mid v \in V \}$ is well defined.  Consider the functional:
$\Gamma: [0,\UpperBound]^V \rightarrow  [0,\UpperBound]^V$, defined by the following set of equations.
\[
    \Gamma(f)(v) =
    \begin{cases}
     
           \sum_{v' \in \post(v)} \delta(v,v')  f(v') & \text{ if } v \in V_\Probabilistic  \\
           \max \{\reward(v)  + f(v') \mid v' \in \post(v) \} & \text{ if } v \in  V_1 \setminus T, \\
           \min \{\reward(v) + f(v') \mid v' \in \post(v) \} & \text{ if } v \in  V_2 \setminus T, \\
           0 & \text{ if } v \in T.
    \end{cases}
\]
Since $([0,\UpperBound]^V, \leq)$ is a complete lattice and $\Gamma$ is monotonic, by the Knaster-Tarski theorems, the (non-empty) set of fixed points of $\Gamma$ forms a complete lattice. From now on, we denote by
$\nu \Gamma$ the greatest fixed point of $\Gamma$.

	The following theorem states that game determinacy is preserved when Player 2 is restricted to play fair strategies. Furthermore, the value of the game is given by the greatest fixed point of $\Gamma$.
\begin{theorem}\label{th:game-determinacy} Let $\StochG$ be a stochastic game stopping under fair strategies. It holds that:
\[\adjustlimits
	\inf_{\strat{2} \in \FairStrats{2}} \sup_{\strat{1} \in \Strategies{1}} \Expect{\strat{1}}{\strat{2}}_{\StochG,v}[\Rewards] = \adjustlimits \sup_{\strat{1} \in \Strategies{1}}   \inf_{\strat{2} \in \FairStrats{2}}  \Expect{\strat{1}}{\strat{2}}_{\StochG,v}[\Rewards] = \nu \Gamma(v)
\]
\end{theorem}
\begin{proof}
  First, note that   $\inf_{\strat{2} \in \DetMemorylessFairStrats{2}}  \sup_{\strat{1} \in \DetMemorylessStrats{1}} \Expect{\strat{1}}{\strat{2}}_{\StochG,v}[\Rewards]$
  is a fixed point of $\Gamma$.  Thus we have:
   \begin{align*}	
        \adjustlimits \sup_{\strat{1} \in \Strategies{1}}   \inf_{\strat{2} \in \FairStrats{2}}  \Expect{\strat{1}}{\strat{2}}_{\StochG,v}[\Rewards]
        & \leq \adjustlimits \inf_{\strat{2} \in \FairStrats{2}} \sup_{\strat{1} \in \Strategies{1}} \Expect{\strat{1}}{\strat{2}}_{\StochG,v}[\Rewards] \\
        &   \leq \adjustlimits  \inf_{\strat{2} \in \DetMemorylessFairStrats{2}}  \sup_{\strat{1} \in \DetMemorylessStrats{1}} \Expect{\strat{1}}{\strat{2}}_{\StochG,v}[\Rewards] 
         \leq   \nu \Gamma(v) 
  \end{align*} 
  for any $v$. The first inequality is a standard property of suprema and infima \cite{Kucera2011}, the second inequality holds because  
  $\DetMemorylessFairStrats{2} \subseteq \FairStrats{2}$ and standard properties of MDPs: by fixing a deterministic memoryless  
  fair strategy for Player $2$ we obtain a transient MDP, the optimal strategy for Player $1$ in this MDP is obtained via a deterministic memoryless strategy \cite{Kallenberg83}, 
  the last inequality holds because  $\inf_{\strat{2} \in \DetMemorylessFairStrats{2}}  \sup_{\strat{1} \in \DetMemorylessStrats{1}} \Expect{\strat{1}}{\strat{2}}_{\StochG,v}[\Rewards]$ is 
  a fixed point of $\Gamma$. 
  
  Rest to prove that $\sup_{\strat{1} \in \Strategies{1}}   \inf_{\strat{2} \in \FairStrats{2}}  \Expect{\strat{1}}{\strat{2}}_{\StochG,v}[\Rewards] \geq \nu \Gamma(v)$. Note that, if there is  $\strat{1} \in \Strategies{1}$ such that
  $\inf_{\strat{2} \in \FairStrats{2}}  \Expect{\strat{1}}{\strat{2}}_{\StochG,v}[\Rewards] \geq \nu \Gamma(v)$ the property above follows by properties of supremum. Consider the strategy $\starredstrat{1}$ defined as follows:
  $\starredstrat{1}(v) \in \argmax \{\nu \Gamma(v') + \reward(v) \mid v' \in \post(v) \}$. Note that $\starredstrat{1}$ is a memoryless strategy. For any memoryless and fair strategy $\strat{2} \in \MemorylessFairStrats{2}$ we have
  $\nu \Gamma(v) \leq \Expect{\starredstrat{1}}{\strat{2}}_{\StochG,v}[\Rewards]$ (by definition of $\Gamma$).  Thus,
  $\nu \Gamma(v) \leq \inf_{\strat{2} \in \MemorylessFairStrats{2}} \Expect{\starredstrat{1}}{\strat{2}}_{\StochG,v}[\Rewards]$
  and therefore:
  $\nu \Gamma(v) \leq \sup_{\strat{1} \in \MemorylessStrats{1}} \inf_{\strat{2} \in \MemorylessFairStrats{1}} \Expect{\strat{1}}{\strat{2}}_{\StochG,v}[\Rewards]$.
  Finally, by Theorem \ref{th:reduce-to-memoryless} we get:
  $\nu \Gamma(v) \leq \sup_{\strat{1} \in \Strategies{1}} \inf_{\strat{2} \in \FairStrats{2}} \Expect{\strat{1}}{\strat{2}}_{\StochG,v}[\Rewards]$.
  \qed
\end{proof}

\paragraph{Considerations for an algorithmic solution.}

Value iteration \cite{Bellman57} has been used to compute maximum/minimum expected accumulated reward in MDPs, e.g. in the {\Prism} model checker.  Usually, the value is computed by approximating the least fixed point from below using the Bellman equations \cite{Bellman57}. In~\cite{DBLP:conf/cav/Baier0L0W17}, the authors propose to approach these values from both a lower and an upper bound (known as interval iteration \cite{DBLP:journals/tcs/HaddadM18}). To do so,  \cite{DBLP:conf/cav/Baier0L0W17} shows a technique for computing upper bounds for the expected total rewards for MDPs.

The above defined functional $\Gamma$ presents a form of Bellman equations that enables a value iteration algorithm to solve our games.  As we are looking for the greatest fixed point, we need to start with some value vector larger than such a fixed point.
If we take our game and fix a fair strategy for the environment, we obtain an MDP. We can then use the techniques presented in \cite{DBLP:conf/cav/Baier0L0W17} to find an upper bound in this MDP, which in turn is an upper bound in the original game. The obvious fair strategy to use is the one based on the uniform distribution (as in Theorem \ref{thm:uniform-prob}).  We have implemented these ideas in a prototype tool as described in the next section.

\section{Experimental Evaluation} \label{sec:experimental_eval}

\definecolor{Gray}{gray}{0.9}
\definecolor{Gray2}{gray}{0.95}
\newcolumntype{g}{>{\columncolor{Gray}}c}
\newcolumntype{h}{>{\columncolor{Gray2}}c}

In order to evaluate the viability of our approach we have  implemented a prototype tool named \textsf{SynthFairy} (available at \cite{SynthFairy})
and run it on two different sets of examples. 
The tool takes as input a model describing the game and the state rewards and returns as output 
the optimal expected total reward  for a given  initial state as well as the synthesized optimal controller strategy (under fairness assumptions). 
The input models are specified using a language resembling {\Prism} notation \cite{DBLP:conf/cav/KwiatkowskaNP11}. The experimental  evaluation shows that our approach can cope with  case studies of reasonable complexity.  For computing these values we set a relative error of at most $\varepsilon = 10^{-6}$.

\paragraph{Roborta vs. the Fair Light.}  Table~\ref{table:resultsRobot} shows the results of the example introduced in Sec.~\ref{sec:mot_example} 
for multiple configurations. 
We considered three variants of the case study: version A (the light does not fail), version B (the light can only fail when trying to signal a green light), and version C (the light can fail when trying to signal any kind of light). 
We assumed that, when Roborta  fails,  she cannot move (this is beneficial to Roborta since she can re-collect the reward);
when the light fails, the robot can freely move into any allowed direction. 
The grid configuration (movement restrictions and rewards) are randomly generated.  For each setting, Table~\ref{table:resultsRobot} describes the results for three different scenarios generated starting at different seeds.  For the grid configuration shown in Sec.~\ref{sec:mot_example} with  parameters $P=0.1$ and $Q=0$, the tool derived the optimal strategy depicted in Fig.~\ref{fig:robot_game_grid} and reports an expected total reward of $5.55$. 



%
\paragraph{Autonomous UAV vs. Human Operator.} We adapted the case study analyzed in  \cite{DBLP:conf/iccps/FengWHT15}. A remotely controlled Unmanned Aerial Vehicle (UAV) is used to perform intelligence, surveillance, and reconnaissance (ISR) missions over a road network. The UAV performs piloting functions autonomously (selecting a path to fly between \emph{waypoints}). The human operator (environment) controls the onboard sensor to capture imagery at a waypoint as well as the piloting functions on certain waypoints (called checkpoints). Note that an operator can continuously try to get a better image by making the UAV loiter around a certain waypoint, this may lead to an unfair behavior.
Each successful capture from an unvisited waypoint grants a reward.
\begin{wrapfigure}[11]{r}{48mm}
\vspace{-11mm}
\centering
\includegraphics[scale=0.60]{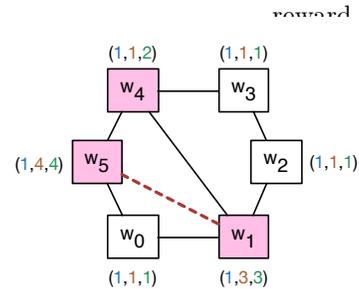}\hspace{1.5em}\mbox{}
\vspace{-10mm}
\caption{UAV Network for ISR missions adapted from \cite{DBLP:conf/iccps/FengWHT15}} \label{fig:uav_game_map}
\end{wrapfigure}
Fig.~\ref{fig:uav_game_map} shows an example of road network consisting of six surveillance waypoints labeled $w_0,w_2,...,w_5$, the edges represent connecting paths, a red-dashed line means that the path is dangerous enough to make the UAV stop working with probability $1$, while on any other path, this probability is $S$. Checkpoints are depicted as pink nodes, therein the operator can still delegate the piloting task to the UAV with probability $D$. Each node is annotated with three possible rewards. For instance,  for $S=0.3$ and $D=0.5$ and the leftmost reward values in each triple, the synthesized strategy for the UAV tries to follow the optimal circuit $w_0,w_1,w_2,w_3,w_4,w_5$. While for  the middle and  rightmost reward values, the optimal circuits to follow are $w_0,w_5,w_0,w_1,w_2,w_3,w_4$ and $w_0,w_5,w_4,w_1,w_2,w_3$, respectively.   Table~\ref{table:resultsUAV} shows the results obtained for  this game for several randomly generated road networks.

All the experiments were run on a MacBook Air with Intel Core i5 at 1.3 GHz and 4 Gb of RAM. The tool and case studies are available in the tool repository.

\begin{table}[t]
  \centering\noindent%
  \caption{Results for Roborta vs. Light Game.  First column describes the grid size.  Second column indicates  the fault probability for the robot ($P$) and light ($Q$). 
  The other columns describe the size of the model, the expected total reward for the optimal strategy, and the computation  time in seconds, respectively,
  for three different randomly generated grid configurations.}
  \label{table:resultsRobot}
  \scalebox{0.9}{
    \begin{tabular}{|c|c|c|c|c|c|g|h|c|g|h|c|}
      \hline
      \multirow{2}{*}{Version} & \multicolumn{2}{c|}{Fault prob.} & \multicolumn{3}{c|}{Size (States/Transitions)} & \multicolumn{3}{c|}{Opt.\ Expect.\ Total Rew.} & \multicolumn{3}{c|}{Time (in sec.)} \\ \hhline{|~|-|-|-|-|-|-|-|-|-|-|-|}
      & $P$ & $Q$ &  \cellcolor{white}s.\ 1 & \cellcolor{white}s.\ 2 & \cellcolor{white}s.\ 3 & \cellcolor{white}\makebox[3.4em][c]{s.\ 1} & \cellcolor{white}\makebox[3.4em][c]{s.\ 2} & \cellcolor{white}\makebox[3.4em][c]{s.\ 3} & \cellcolor{white}\makebox[1.8em][c]{s.\ 1} & \cellcolor{white}\makebox[1.8em][c]{s.\ 2} & \cellcolor{white}\makebox[1.8em][c]{s.\ 3} \\  
      \hline
       \multirow{2}{3em}{\centering $A$ \\ $60{\times}8$}
       & $0.1$ & $-$ & \multirow{2}{*}{\centering $\displaystyle \begin{array}{c} \text{st. }2614 \\ \text{tr. }4418 \end{array}$} & \multirow{2}{*}{\centering $\displaystyle \begin{array}{c} \text{st. }2540 \\ \text{tr. }4266 \end{array}$}  & \multirow{2}{*}{\centering $\displaystyle \begin{array}{c} \text{st. }2552 \\ \text{tr. }4295 \end{array}$}  & $26.66$ & $31.11$ & $27.77$ & $10$ & $9$ & $12$\\ \hhline{|~|-|-|~|~|~|-|-|-|-|-|-|}
       & $0.5$ & $-$ & & & & $48$ & $56$ & $50$ & $14$ & $16$ & $12$\\ \hline

       \multirow{2}{3em}{\centering $A$ \\ $120{\times}16$}
       & $0.1$ & $-$ & \multirow{2}{*}{\centering $\displaystyle \begin{array}{c} \text{st. }10161 \\ \text{tr. }17125 \end{array}$} & \multirow{2}{*}{\centering $\displaystyle \begin{array}{c} \text{st. }10217 \\ \text{tr. }17223 \end{array}$}  & \multirow{2}{*}{\centering $\displaystyle \begin{array}{c} \text{st. }10249 \\ \text{tr. }17319 \end{array}$}  & $62.22$ & $55.55$ & $48.88$ & $99$ & $116$ & $96$\\ \hhline{|~|-|-|~|~|~|-|-|-|-|-|-|}
       & $0.5$ & $-$ & & & & $112$ & $100$ & $88$ & $133$ & $149$ & $137$\\ \hline

       \multirow{4}{3em}{\centering $B$ \\ $60{\times}8$}
       & \multirow{2}{*}{$0.1$} & $0.1$ & \multirow{4}{*}{\centering $\displaystyle \begin{array}{c} \text{st. }4724 \\ \text{tr. }8788 \end{array}$} & \multirow{4}{*}{\centering $\displaystyle \begin{array}{c} \text{st. }4586 \\ \text{tr. }8484 \end{array}$}  & \multirow{4}{*}{\centering $\displaystyle \begin{array}{c} \text{st. }4609 \\ \text{tr. }8542 \end{array}$}  & $42.6$ & $44.59$ & $42.23$ & $20$ & $24$ & $21$\\ \hhline{|~|~|-|~|~|~|-|-|-|-|-|-|}
       & & $0.5$ & & & & $130.14$ & $127.7$ & $136.22$ & $32$ & $33$ & $30$\\ \hhline{|~|-|-|~|~|~|-|-|-|-|-|-|}
       & \multirow{2}{*}{$0.5$} & $0.1$ & & & & $76.68$ & $80.26$ & $76.02$ & $32$ & $31$ & $32$\\ \hhline{|~|~|-|~|~|~|-|-|-|-|-|-|}
       & & $0.5$ & & & & $234.26$ & $229.87$ & $245.21$ & $53$ & $50$ & $56$\\ \hline

       \multirow{4}{3em}{\centering $B$ \\ $120{\times}16$}
       & \multirow{2}{*}{$0.1$} & $0.1$ & \multirow{4}{*}{\centering $\displaystyle \begin{array}{c} \text{st. }18384 \\ \text{tr. }34154 \end{array}$} & \multirow{4}{*}{\centering $\displaystyle \begin{array}{c} \text{st. }18486 \\ \text{tr. }34350 \end{array}$}  & \multirow{4}{*}{\centering $\displaystyle \begin{array}{c} \text{st. }18550 \\ \text{tr. }34542 \end{array}$}  & $91.19$ & $87.27$ & $80.07$ & $260$ & $250$ & $277$\\ \hhline{|~|~|-|~|~|~|-|-|-|-|-|-|}
       & & $0.5$ & & & & $281.83$ & $281.48$ & $265.33$ & $375$ & $366$ & $359$\\ \hhline{|~|-|-|~|~|~|-|-|-|-|-|-|}
       & \multirow{2}{*}{$0.5$} & $0.1$ & & & & $164.15$ & $157.1$ & $144.13$ & $389$ & $378$ & $376$\\ \hhline{|~|~|-|~|~|~|-|-|-|-|-|-|}
       & & $0.5$ & & & & $507.30$ & $506.67$ & $477.6$ & $555$ & $553$ & $535$\\ \hline

       \multirow{4}{3em}{\centering $C$ \\ $60{\times}8$}
       & \multirow{2}{*}{$0.1$} & $0.1$ & \multirow{4}{*}{\centering $\displaystyle \begin{array}{c} \text{st. }5212 \\ \text{tr. }9764 \end{array}$} & \multirow{4}{*}{\centering $\displaystyle \begin{array}{c} \text{st. }5064 \\ \text{tr. }9440 \end{array}$}  & \multirow{4}{*}{\centering $\displaystyle \begin{array}{c} \text{st. }5088 \\ \text{tr. }9500 \end{array}$}  & $46.32$ & $47.07$ & $44.87$ & $23$ & $22$ & $23$\\ \hhline{|~|~|-|~|~|~|-|-|-|-|-|-|}
       & & $0.5$ & & & & $143.35$ & $146.41$ & $153.98$ & $34$ & $32$ & $34$\\ \hhline{|~|-|-|~|~|~|-|-|-|-|-|-|}
       & \multirow{2}{*}{$0.5$} & $0.1$ & & & & $83.37$ & $84.73$ & $80.77$ & $36$ & $34$ & $34$\\ \hhline{|~|~|-|~|~|~|-|-|-|-|-|-|}
       & & $0.5$ & & & & $258.04$ & $263.53$ & $277.17$ & $54$ & $51$ & $52$\\ \hline

       \multirow{4}{3em}{\centering $C$ \\ $120{\times}16$}
       & \multirow{2}{*}{$0.1$} & $0.1$ & \multirow{4}{*}{\centering $\displaystyle \begin{array}{c} \text{st. }20290 \\ \text{tr. }37966 \end{array}$} & \multirow{4}{*}{\centering $\displaystyle \begin{array}{c} \text{st. }20402 \\ \text{tr. }38182 \end{array}$}  & \multirow{4}{*}{\centering $\displaystyle \begin{array}{c} \text{st. }20466 \\ \text{tr. }38374 \end{array}$}  & $98.25$ & $93.74$ & $88.33$ & $297$ & $297$ & $331$\\ \hhline{|~|~|-|~|~|~|-|-|-|-|-|-|}
       & & $0.5$ & & & & $321.18$ & $317.61$ & $311.62$ & $405$ & $403$ & $496$\\ \hhline{|~|-|-|~|~|~|-|-|-|-|-|-|}
       & \multirow{2}{*}{$0.5$} & $0.1$ & & & & $176.85$ & $168.73$ & $158.99$ & $440$ & $426$ & $527$\\ \hhline{|~|~|-|~|~|~|-|-|-|-|-|-|}
       & & $0.5$ & & & & $578.13$ & $571.71$ & $560.92$ & $606$ & $640$ & $723$\\ 

      \hline
    \end{tabular}
  }
\end{table}

\begin{table}[t]
  \centering\noindent%
  \caption{Results for the UAV vs. Operator Game. First column describes the number of waypoints used.  Second column indicates probability of delegation ($D$), and the probability that the UAV stops working ($S$). 
  The other columns show the size of the model, the expected total reward for the optimal strategy, and the computation time in seconds, respectively, for three different randomly generated roadmap configurations.}
  \label{table:resultsUAV}
  \scalebox{0.9}{
    \begin{tabular}{|c|c|c|c|c|c|g|h|c|g|h|c|}
      \hline
      \multirow{2}{*}{Version} & \multicolumn{2}{c|}{Prob.} & \multicolumn{3}{c|}{Size(States/Transitions)} & \multicolumn{3}{c|}{Opt.\ Expect.\ Total Rew.} & \multicolumn{3}{c|}{Time (in sec.)} \\ \hhline{|~|-|-|-|-|-|-|-|-|-|-|-|}
      & $D$ & $S$ &  \cellcolor{white}s.\ 1 & \cellcolor{white}s.\ 2 & \cellcolor{white}s.\ 3 & \cellcolor{white}\makebox[3.4em][c]{s.\ 1} & \cellcolor{white}\makebox[3.4em][c]{s.\ 2} & \cellcolor{white}\makebox[3.4em][c]{s.\ 3} & \cellcolor{white}\makebox[1.8em][c]{s.\ 1} & \cellcolor{white}\makebox[1.8em][c]{s.\ 2} & \cellcolor{white}\makebox[1.8em][c]{s.\ 3} \\  
      \hline
       \multirow{4}{3em}{\centering UAV \\ $6w.$}
       & \multirow{2}{*}{$0.1$} & $0.05$ & \multirow{4}{*}{\centering $\displaystyle \begin{array}{c} \text{st. }501 \\ \text{tr. }857 \end{array}$} & \multirow{4}{*}{\centering $\displaystyle \begin{array}{c} \text{st. }1148 \\ \text{tr. }2166 \end{array}$} & \multirow{4}{*}{\centering $\displaystyle \begin{array}{c} \text{st. }302 \\ \text{tr. }508 \end{array}$} & $16.72$ & $12.47$ & $13.14$ & $8$ & $25$ & $1$\\ \hhline{|~|~|-|~|~|~|-|-|-|-|-|-|}
       & & $0.1$ & & & & $15.73$ & $11.15$ & $12.63$ & $5$ & $14$ & $2$\\ \hhline{|~|-|-|~|~|~|-|-|-|-|-|-|}
       & \multirow{2}{*}{$0.5$} & $0.05$ & & & & $20.49$ & $12.77$ & $17.05$ & $7$ & $22$ & $3$\\ \hhline{|~|~|-|~|~|~|-|-|-|-|-|-|}
       & & $0.1$ & & & & $18.8$ & $11.67$ & $15.95$ & $7$ & $13$ & $2$\\ \hline

       \multirow{4}{3em}{\centering UAV \\ $8w.$}
       & \multirow{2}{*}{$0.1$} & $0.05$ & \multirow{4}{*}{\centering $\displaystyle \begin{array}{c} \text{st. }5335 \\ \text{tr. }9734 \end{array}$} & \multirow{4}{*}{\centering $\displaystyle \begin{array}{c} \text{st. }8289 \\ \text{tr. }15642 \end{array}$} & \multirow{4}{*}{\centering $\displaystyle \begin{array}{c} \text{st. }3167 \\ \text{tr. }5723 \end{array}$} & $17.88$ & $40.59$ & $24.6$ & $47$ & $183$ & $51$\\ \hhline{|~|~|-|~|~|~|-|-|-|-|-|-|}
       & & $0.1$ & & & & $17.11$ & $34.3$ & $21.48$ & $63$ & $128$ & $31$\\ \hhline{|~|-|-|~|~|~|-|-|-|-|-|-|}
       & \multirow{2}{*}{$0.5$} & $0.05$ & & & & $26$ & $42.21$ & $30.87$ & $67$ & $209$ & $41$\\ \hhline{|~|~|-|~|~|~|-|-|-|-|-|-|}
       & & $0.1$ & & & & $23.44$ & $36.08$ & $24.72$ & $51$ & $114$ & $27$\\ \hline

       \multirow{4}{3em}{\centering UAV \\ $10w.$}
       & \multirow{2}{*}{$0.1$} & $0.05$ & \multirow{4}{*}{\centering $\displaystyle \begin{array}{c} \text{st. }15775 \\ \text{tr. }28257 \end{array}$} & \multirow{4}{*}{\centering $\displaystyle \begin{array}{c} \text{st. }11212 \\ \text{tr. }20624 \end{array}$} & \multirow{4}{*}{\centering $\displaystyle \begin{array}{c} \text{st. }21153 \\ \text{tr. }39348 \end{array}$} & $39.76$ & $28.7$ & $19.76$ & $659$ & $656$ & $892$\\ \hhline{|~|~|-|~|~|~|-|-|-|-|-|-|}
       & & $0.1$ & & & & $35.43$ & $23.36$ & $16.2$ & $389$ & $261$ & $586$\\ \hhline{|~|-|-|~|~|~|-|-|-|-|-|-|}
       & \multirow{2}{*}{$0.5$} & $0.05$ & & & & $42.13$ & $30.77$ & $24.56$ & $644$ & $421$ & $891$\\ \hhline{|~|~|-|~|~|~|-|-|-|-|-|-|}
       & & $0.1$ & & & & $37.11$ & $26.08$ & $19.27$ & $412$ & $261$ & $592$\\ \hline

      \hline
    \end{tabular}
  }
\end{table}

\section{Related Work} \label{sec:related_work}
	Stochastic games with payoff functions have been extensively investigated in the literature. In \cite{FilarV96}, several results are presented about \emph{transient games},
a generalized version of stopping stochastic games with total reward payoff. 
In transient games, both players possess optimal (memoryless and deterministic) strategies. 
Most importantly, the games are determined and their value can be computed as the least fixed point of a set of equations. 
Most of these results are based on the fact that the $\Gamma$ functional (see Section \ref{sec:fair-strats}) for transient games has a unique fixed point.
Notice that in this paper we have dealt with games that are stopping only under fairness assumptions. Thus, the corresponding functional 
may have several fixed points. Hence, the main results presented in \cite{FilarV96} do not apply to our setting.

\cite{DBLP:journals/fmsd/ChenFKPS13} and \cite{SvorenovaKwiatkowska16} present logical frameworks for the verification and synthesis of systems.  While~\cite{DBLP:journals/fmsd/ChenFKPS13} provides a solution for a  probabilistic branching temporal logic extended with expected total, discounted, and average reward objective functions, \cite{SvorenovaKwiatkowska16} does the same in a similar extension of a probabilistic linear temporal logic. Both frameworks were implemented in the tool \Prism~\cite{DBLP:conf/cav/KwiatkowskaNP11}. Although a vast class of properties can be expressed in these frameworks, none of them are presented under fair environments.  In fact, these works are on stochastic multiplayer games in which each player is treated equally.

	\emph{Stochastic shortest path games} \cite{PatekBertsekas99} are two-player stochastic games with (negative or positive) rewards in which the minimizer's strategies are classified into \emph{proper} and \emph{improper},
proper strategies are those ensuring termination.  As proven in \cite{PatekBertsekas99}, these games are determined, and both players posses memoryless optimal strategies.  To prove these results,  the authors assume that the expected game value for improper strategies is $\infty$, this ensures that the corresponding functional is a contraction and thus it has a unique fixpoint. In contrast, we restrict ourselves to non-negative rewards but we do not make any assumptions over unfair strategies, as mentioned above the corresponding functional for our games may have several fixpoints. Furthermore, we proved that the value of the game is given by the greatest fixpoint of $\Gamma$.  In recent years, several authors have investigated stochastic shortest path problems for MDPs (i.e,  one-player games), where the assumption over improper strategies  is relaxed (e.g., \cite{DBLP:conf/lics/Baier0DGS18}); to the best of our knowledge, these results have not be extended to two-player games.  

	
	In \cite{DBLP:conf/ifipTCS/BolligC04} the authors tackle the problem of synthesizing a controller that maximizes the probability of satisfying an {\LTL} property. Fairness strategies are used to reduce this problem to the synthesis of a  controller maximizing a {\PCTL} property over a product game. However, this article does not address expected rewards and game determinacy under fairness assumptions. 
	
	Interestingly, in \cite{DBLP:conf/fossacs/AsarinCV10} the authors consider the problem of winning a (non-stochastic) two-player games with fairness assumptions over the environment. The objective of the system is to guarantee an $\omega$-regular property. The authors show that winning in these games is equivalent to almost-sure winning in a Markov decision process. It must be noted that this work only considers non-stochastic games.  Furthermore, payoff functions are not considered therein.

\section{Concluding remarks} \label{sec:conclusions}
In this paper, we have investigated the properties of stochastic games with total reward payoff under the assumption that the minimizer (i.e., the environment) plays only with fair strategies.  
We have shown that, in this scenario, determinacy is preserved and  both players have optimal memoryless and deterministic strategies; furthermore,  the value of the game can be calculated by approximating a greatest fixed point of a Bellman operator. We have only considered non-negative rewards in this paper. A possible way of extending the results presented here to games with negative rewards is to adapt the techniques presented in \cite{DBLP:conf/lics/Baier0DGS18} for MDPs with negative costs, we leave this as a further work.

 In order to show the applicability of our technique, we have presented two examples of applications and an experimental evaluation over diverse instances of these case studies using our prototype tool. We believe that fairness assumptions allow one to consider more realistic behavior of the environment. 
As future work, we plan to build a more robust tool and thus model and evaluate other case studies.
	 
We have not investigated other common payoff functions such as discounted payoff or limiting-average payoff. A benefit of these classes of functions is that the value of games are well-defined even when the games are not stopping.  
At first sight, the notion of fairness is little relevant for games with discounted payoff, since these kinds of payoff functions take most of their value from the initial parts of runs. For limiting-average the situation is different, and fairness assumptions may be relevant as they could change the value of games, we leave this as further work.



\bibliography{references}

\begin{thebibliography}{10}

\bibitem{Altman99}
Eitan Altman.
\newblock {\em Constrained Markov Decision Processes}.
\newblock Chapman and Hall/CRC, 1999.

\bibitem{DBLP:conf/fossacs/AsarinCV10}
Eugene Asarin, Rapha{\"{e}}l Chane{-}Yack{-}Fa, and Daniele Varacca.
\newblock Fair adversaries and randomization in two-player games.
\newblock In C.{-}H.~Luke Ong, editor, {\em Foundations of Software Science and
  Computational Structures, 13th International Conference, {FOSSACS} 2010, Held
  as Part of the Joint European Conferences on Theory and Practice of Software,
  {ETAPS} 2010, Paphos, Cyprus, March 20-28, 2010. Proceedings}, volume 6014 of
  {\em Lecture Notes in Computer Science}, pages 64--78. Springer, 2010.
\newblock \href {https://doi.org/10.1007/978-3-642-12032-9_6}
  {\path{doi:10.1007/978-3-642-12032-9_6}}.

\bibitem{DBLP:journals/moc/Azimzadeh19}
Parsiad Azimzadeh.
\newblock A fast and stable test to check if a weakly diagonally dominant
  matrix is a nonsingular m-matrix.
\newblock {\em Math. Comput.}, 88(316):783--800, 2019.
\newblock \href {https://doi.org/10.1090/mcom/3347}
  {\path{doi:10.1090/mcom/3347}}.

\bibitem{DBLP:conf/lics/Baier0DGS18}
Christel Baier, Nathalie Bertrand, Clemens Dubslaff, Daniel Gburek, and Ocan
  Sankur.
\newblock Stochastic shortest paths and weight-bounded properties in markov
  decision processes.
\newblock In Anuj Dawar and Erich Gr{\"{a}}del, editors, {\em Proceedings of
  the 33rd Annual {ACM/IEEE} Symposium on Logic in Computer Science, {LICS}
  2018, Oxford, UK, July 09-12, 2018}, pages 86--94. {ACM}, 2018.
\newblock \href {https://doi.org/10.1145/3209108.3209184}
  {\path{doi:10.1145/3209108.3209184}}.

\bibitem{DBLP:conf/ifipTCS/BolligC04}
Christel Baier, Marcus Gr{\"{o}}{\ss}er, Martin Leucker, Benedikt Bollig, and
  Frank Ciesinski.
\newblock Controller synthesis for probabilistic systems.
\newblock In Jean{-}Jacques L{\'{e}}vy, Ernst~W. Mayr, and John~C. Mitchell,
  editors, {\em Exploring New Frontiers of Theoretical Informatics, {IFIP} 18th
  World Computer Congress, {TC1} 3rd International Conference on Theoretical
  Computer Science (TCS2004), 22-27 August 2004, Toulouse, France}, volume 155
  of {\em {IFIP}}, pages 493--506. Kluwer/Springer, 2004.
\newblock \href {https://doi.org/10.1007/1-4020-8141-3_38}
  {\path{doi:10.1007/1-4020-8141-3_38}}.

\bibitem{BaierK08}
Christel Baier and Joost-Pieter Katoen.
\newblock {\em Principles of Model Checking}.
\newblock The MIT Press, 2008.

\bibitem{DBLP:conf/cav/Baier0L0W17}
Christel Baier, Joachim Klein, Linda Leuschner, David Parker, and Sascha
  Wunderlich.
\newblock Ensuring the reliability of your model checker: Interval iteration
  for markov decision processes.
\newblock In Rupak Majumdar and Viktor Kuncak, editors, {\em Computer Aided
  Verification - 29th International Conference, {CAV} 2017, Heidelberg,
  Germany, July 24-28, 2017, Proceedings, Part {I}}, volume 10426 of {\em
  Lecture Notes in Computer Science}, pages 160--180. Springer, 2017.
\newblock \href {https://doi.org/10.1007/978-3-319-63387-9_8}
  {\path{doi:10.1007/978-3-319-63387-9_8}}.

\bibitem{DBLP:journals/dc/BaierK98}
Christel Baier and Marta~Z. Kwiatkowska.
\newblock Model checking for a probabilistic branching time logic with
  fairness.
\newblock {\em Distributed Comput.}, 11(3):125--155, 1998.
\newblock \href {https://doi.org/10.1007/s004460050046}
  {\path{doi:10.1007/s004460050046}}.

\bibitem{Bellman57}
Richard Bellman.
\newblock {\em Dynamic Programming}.
\newblock Princeton University Press, Princeton, NJ, USA, 1 edition, 1957.

\bibitem{ChatterjeeH12}
Krishnendu Chatterjee and Thomas~A. Henzinger.
\newblock A survey of stochastic {\(\omega\)}-regular games.
\newblock {\em J. Comput. Syst. Sci.}, 78(2):394--413, 2012.

\bibitem{DBLP:journals/fmsd/ChenFKPS13}
Taolue Chen, Vojtech Forejt, Marta~Z. Kwiatkowska, David Parker, and Aistis
  Simaitis.
\newblock Automatic verification of competitive stochastic systems.
\newblock {\em Formal Methods Syst. Des.}, 43(1):61--92, 2013.
\newblock \href {https://doi.org/10.1007/s10703-013-0183-7}
  {\path{doi:10.1007/s10703-013-0183-7}}.

\bibitem{Condon90}
Anne Condon.
\newblock On algorithms for simple stochastic games.
\newblock In Jin{-}Yi Cai, editor, {\em Advances In Computational Complexity
  Theory, Proceedings of a {DIMACS} Workshop, New Jersey, USA, December 3-7,
  1990}, volume~13 of {\em {DIMACS} Series in Discrete Mathematics and
  Theoretical Computer Science}, pages 51--71. {DIMACS/AMS}, 1990.

\bibitem{Condon92}
Anne Condon.
\newblock The complexity of stochastic games.
\newblock {\em Inf. Comput.}, 96(2):203--224, 1992.

\bibitem{DBLP:conf/icse/DIppolitoBPU11}
Nicol{\'{a}}s D'Ippolito, V{\'{\i}}ctor~A. Braberman, Nir Piterman, and
  Sebasti{\'{a}}n Uchitel.
\newblock Synthesis of live behaviour models for fallible domains.
\newblock In Richard~N. Taylor, Harald~C. Gall, and Nenad Medvidovic, editors,
  {\em Proceedings of the 33rd International Conference on Software
  Engineering, {ICSE} 2011, Waikiki, Honolulu , HI, USA, May 21-28, 2011},
  pages 211--220. {ACM}, 2011.
\newblock \href {https://doi.org/10.1145/1985793.1985823}
  {\path{doi:10.1145/1985793.1985823}}.

\bibitem{DBLP:conf/iccps/FengWHT15}
Lu~Feng, Clemens Wiltsche, Laura~R. Humphrey, and Ufuk Topcu.
\newblock Controller synthesis for autonomous systems interacting with human
  operators.
\newblock In Alexandre~M. Bayen and Michael~S. Branicky, editors, {\em
  Proceedings of the {ACM/IEEE} Sixth International Conference on
  Cyber-Physical Systems, {ICCPS} 2015, Seattle, WA, USA, April 14-16, 2015},
  pages 70--79. {ACM}, 2015.
\newblock \href {https://doi.org/10.1145/2735960.2735973}
  {\path{doi:10.1145/2735960.2735973}}.

\bibitem{FilarV96}
Jerzy Filar and Koos Vrieze.
\newblock {\em Competitive Markov Decision Processes}.
\newblock Springer-Verlag, Berlin, Heidelberg, 1996.
\newblock \href {https://doi.org/10.1007/978-1-4612-4054-9}
  {\path{doi:10.1007/978-1-4612-4054-9}}.

\bibitem{DBLP:journals/tcs/HaddadM18}
Serge Haddad and Benjamin Monmege.
\newblock Interval iteration algorithm for {MDPs} and {IMDPs}.
\newblock {\em Theor. Comput. Sci.}, 735:111--131, 2018.
\newblock \href {https://doi.org/10.1016/j.tcs.2016.12.003}
  {\path{doi:10.1016/j.tcs.2016.12.003}}.

\bibitem{Kallenberg83}
Lodewijk Kallenberg.
\newblock {\em Linear Programming and Finite Markovian Control Problems}.
\newblock Mathematisch Centrum, Amsterdam, Netherland, 1983.

\bibitem{Kucera2011}
Anton{\'{\i}}n Ku{\v{c}}era.
\newblock Turn-based stochastic games.
\newblock In Krzysztof~R. Apt and Erich Gr{\"{a}}del, editors, {\em Lectures in
  Game Theory for Computer Scientists}, pages 146--184. Cambridge University
  Press, 2011.
\newblock \href {https://doi.org/10.1017/CBO9780511973468.006}
  {\path{doi:10.1017/CBO9780511973468.006}}.

\bibitem{DBLP:conf/cav/KwiatkowskaNP11}
Marta~Z. Kwiatkowska, Gethin Norman, and David Parker.
\newblock {PRISM} 4.0: Verification of probabilistic real-time systems.
\newblock In Ganesh Gopalakrishnan and Shaz Qadeer, editors, {\em Computer
  Aided Verification - 23rd International Conference, {CAV} 2011, Snowbird, UT,
  USA, July 14-20, 2011. Proceedings}, volume 6806 of {\em Lecture Notes in
  Computer Science}, pages 585--591. Springer, 2011.
\newblock \href {https://doi.org/10.1007/978-3-642-22110-1_47}
  {\path{doi:10.1007/978-3-642-22110-1_47}}.

\bibitem{Martin98}
Donald~A. Martin.
\newblock The determinacy of {B}lackwell games.
\newblock {\em J. Symb. Log.}, 63(4):1565--1581, 1998.

\bibitem{MorgensternNeuman42}
Oskar Morgenstern and John von Neumann.
\newblock {\em Theory of Games and Economic Behavior}.
\newblock Princeton University Press, 1st edition, 1942.

\bibitem{PatekBertsekas99}
Stephen~D. Patek and Dimitri~P. Bertsekas.
\newblock Stochastic shortest path games.
\newblock {\em SIAM Journal on Control and Optimization}, 37:804--824, 1999.

\bibitem{Shapley1095}
Lloyd Shapley.
\newblock Stochastic games.
\newblock {\em Proceedings of the National Academy of Sciences},
  39(10):1095--1100, 1953.
\newblock \href {https://doi.org/10.1073/pnas.39.10.1095}
  {\path{doi:10.1073/pnas.39.10.1095}}.

\bibitem{SvorenovaKwiatkowska16}
M{\'a}ria Svore{\v n}ov{\'a} and Marta~Z. Kwiatkowska.
\newblock Quantitative verification and strategy synthesis for stochastic
  games.
\newblock {\em European Journal of Control}, 2016.

\bibitem{SynthFairy}
{SynthFairy}: {C}ontroller synthesis tool.
\newblock \url{https://github.com/cl-unrc-lab/SynthFairy}.

\end{thebibliography}


\appendix

\spnewtheorem*{proofofclaim}{Proof of claim}{\itshape}{\rmfamily}
\newenvironment{myproof}[1][]{\begingroup\renewcommand{\proofname}{\textbf{Proof #1}}\begin{proof}}{\end{proof}\endgroup}

\section{Proofs}\label{sec:appendix}

\begin{myproof}[of Lemma~\ref{lemma:memoryless-strat}]
  Fix strategies $\strat{1} \in \Strategies{1}$, $\strat{2} \in
  \FairStrats{2}$ and a vertex $v \in V$ such that
  $\Prob{\strat{1}}{\strat{2}}_{\mathcal{G}, v}(\Diamond T) < 1$.
  Note that we can define an MDP (named $\StochG^{V_1{+}V_2}$) from
  $\StochG$ by considering that $V_1$ and $V_2$ belong to the same
  player.  For this MDP consider the set $U = \{ (V',\delta')
  \mid(V',\delta') \in \EndComp(\StochG^{V_1{+}V_2}) \text{ and }
  V'\cap T = \emptyset \}$ of end components not containing terminal
  vertices.  So, we define a strategy for this player combining
  $\strat{1}$ and $\strat{2}$ (named $\strat{1}{+}\strat{2}$) as
  follows: it behaves as $\strat{1}$ in finite paths ending in a state
  of $V_1$, and behaves as $\strat{2}$ in finite paths ending in a
  state of $V_2$. Note that we have
  $\MDPProb{\strat{1}{+}\strat{2}}_{\StochG^{V_1+V_2},v}(\Diamond T) <
  1$. By the limit properties of MDPs (Theorem 10.120 \cite{BaierK08})
  we have that:
  \begin{equation}\label{lemma:memoryless-strat-eq1}
    \MDPProb{\strat{1}{+}\strat{2}}_{\StochG^{V_1+V_2},v} \left(\{ \omega \in \MDPPaths{\strat{1}{+}\strat{2}}_{\StochG^{V_1+V_2},v} \mid \limit(\omega)  \in U \}\right) > 0,
  \end{equation}
  where $\limit(\omega)$ denotes the end-component that is repeated
  infinitely often in $\omega$, as defined in \cite{BaierK08}.  Note
  that, the fairness of $\strat{2}$ and
  (\ref{lemma:memoryless-strat-eq1}) imply that there is a reachable
  end component $\EC{C}= (V',\delta') \in U$ such that, for all
  $v' \in V_2 \cap V'$, $\post^\delta(v') \subseteq V'$ $(\dag)$.
  This
  can be proven by contradiction: if $(\dag)$ does not hold we have
  that for every end component in $U$ there must be some vertex in
  $V_2$ such that some of its successors is not in the component; but,
  since $\strat{2}$ is fair, we should have
  $\MDPProb{\strat{1}{+}\strat{2}}_{\StochG^{V_1+V_2},v} \left ( \{ \omega \in \MDPPaths{\strat{1}{+}\strat{2}}_{\StochG^{V_1+V_2},v} \mid \limit(\omega) \in U \}\right) = 0$
  contradicting (\ref{lemma:memoryless-strat-eq1}). Thus $(\dag)$ must hold.

  Hence, there is some finite path $\hat{\omega} = v_0 v_1 v_2 \dots v_k$
  in this MDP such that $v_i \notin T$ for every $i$, $v_k \in V'$, and
  $\MDPProb{\strat{1}{+}\strat{2}}_{\StochG^{V_1+V_2},v}(v_0\dots v_k) > 0$.

  Now, we define $\strat{1}'$ as follows:
  $\strat{1}'(v') = \hat{\omega}_{i+1}$, if $i < k$ is the largest
  index such that $v' = \hat{\omega}_{i}$;
  $\strat{1}'(v') = v''$ for some arbitrary $v'' \in
  V'\cap\post^\delta(v)$, if $v' \in V'$;
  otherwise $\strat{1}'(v') = v''$ for some arbitrary
  $v'' \in \post^\delta(v')$.
  $\strat{2}'$ is defined as the uniform strategy, that is:
  $\strat{2}'(v')(v'') = \frac{1}{\# \post^\delta(v')}$, for every $v'
  \in V_2$ and $v'' \in \post^\delta(v')$.
  
  $\strat{1}'$ is defined such that it jumps ahead on $\hat{\omega}$
  skipping all possible loops introduced by $\strat{1}$.  Let
  $\hat{\omega}{\downharpoonleft}_{\strat{1}'}$ be the finite path
  obtained by following $\hat{\omega}$ skipping all loops according
  to $\strat{1}'$.
  Note that $\hat{\omega}{\downharpoonleft}_{\strat{1}'}$ is a valid
  path in $\StochG$, that it still ends in state $v_k$, and that
  $\MDPProb{\strat{1}',\strat{2}'}_{\StochG,v}({\hat{\omega}{\downharpoonleft}_{\strat{1}'}})>0$.
  Moreover, by ($\dag$) and the definition of $\strat{1}'$, we have
  that $\MDPProb{\strat{1}',\strat{2}'}_{\StochG, v_k}(\Box V') = 1$.
  Thus $\MDPProb{\strat{1}',\strat{2}'}_{\StochG, v}(\Diamond T) < 1$.
  Furthermore, note that $\strat{1}'$ is memoryless and $\strat{2}'$
  is fair.
\qed
\end{myproof}

\begin{myproof}[of Theorem~\ref{thm:uniform-prob}]
``If'': Suppose that $\Prob{\strat{1}}{\uniformstrat{2}}_{\StochG,v}(\Diamond T)=1$ for every Player~1's strategy $\strat{1}$.  Moreover, 
assume, for the sake of contradiction, that
$\Prob{\strat{1}'}{\strat{2}'}_{\StochG,v}(\Diamond T) < 1$ for some  $\strat{1}'\in \Strategies{1}$ and $\strat{2}'\in \FairStrats{2}$. 
By Lemma \ref{lemma:memoryless-strat}, we can safely assume that $\strat{1}'$ is memoryless and deterministic.
Thus, by fixing $\strat{1}'$, we obtain a (finite) Markov decision process denoted $\StochG^{\strat{1}'}$.
By assumption, we have that
$\inf_{\strat{2} \in \FairStrats{2}} \MDPProb{\strat{2}}_{\StochG^{\strat{1}'},v}(\Diamond T) < 1 $. In addition, by Theorem 10.133 in \cite{BaierK08}, we have that $\inf_{ \strat{2} \in \FairStrats{2}}\MDPProb{\strat{2}}_{\StochG^{\strat{1}'},v}(\Diamond T) = 1 - \sup_{\strat{2} \in \FairStrats{2}} \MDPProb{\strat{2}}_{\StochG^{\strat{1}'},v}(\neg T \U U)$, where $T$ is the set of terminal states,  $\neg T$ is its complement, and
$U = \bigcup \{ V' \mid (V', \delta') \in \EndComp(\StochG^{\strat{1}'}) \text{ and }  T \cap V' = \emptyset\}$.
Let $\strat{2}''$ be the fair strategy that maximizes the value of $\MDPProb{\strat{2}''}_{\StochG^{\strat{1}'},v}(\neg T \U U)$, which exists as shown in the proof of Theorem 10.133 in \cite{BaierK08}.
Then we have that $\MDPProb{\strat{2}''}_{\StochG^{\strat{1}'},v}(\neg T \U U) > 0$ and hence
there is a path $v_0 v_1 \dots v_n$ with positive probability such that $v_i \notin T$ and $v_n$ belongs to an end-component $\EC{C} = (V'', \delta'')$
such that $T \cap V'' = \emptyset$. Since $\strat{2}''$ is fair, by our definition of fair strategy, the end component $\EC{C}$ contains all the 
transitions in $\StochG^{\strat{1}'}$ when restricted  to $V''$.
Therefore, $\MDPProb{\uniformstrat{2}}_{\StochG^{\strat{1}'},v}(\Diamond V'') > 0$, and 
$\MDPProb{\uniformstrat{2}}_{\StochG^{\strat{1}'},v''}(\Box V'') = 1 $ 
from any state in $v'' \in V''$, which implies that $\MDPProb{\uniformstrat{2}}_{\StochG{\strat{1}'},v}(\neg T \U V'') > 0$. 
Hence, $\MDPProb{\uniformstrat{2}}_{\StochG^{\strat{1}'},v}(\Diamond T) < 1$, which contradicts our initial assumption.

\noindent ``Only If'': This part is direct since $\uniformstrat{2}$ is a fair strategy.
\qed
\end{myproof} 

\begin{myproof}[of Theorem~\ref{thm:stopping-algorithm}]
  (2) is an immediate consequence of (1). So, we only prove the first
  statement.

  Notice that
  $\StochG^{\uniformstrat{2}} = (V, (V_1, \emptyset, {V_2 \cup V_\Probabilistic}), \delta')$,
  where $\delta'(v,\cdot)=\delta(v,\cdot)$ if $v\in V_1\cup V_\Probabilistic$
  and $\delta'(v,\cdot)=\uniformstrat{2}(v)$ if $v\in V_2$.
  On this MDP, operators $\Apre$ and $\Epre$ are defined by
  \begin{align*}
    \Epre(C) \ = \ {}&\{ v \in V \mid \delta'(v,C) > 0\} \\
    \Apre(C) \ = \ {}&\{ v \in V_2\cup V_\Probabilistic \mid \delta'(v,C) > 0\} \\
		     &\cup \ \{ v \in  V_1 \mid \forall v' \in V : \delta'(v,v') > 0 \Rightarrow v' \in C \}
  \end{align*}
  It is straightforward to verify that $\EFairpre(C)=\Epre(C)$ and
  $\AFairpre(C)=\Apre(C)$ for any $C\subseteq V$.

  As a consequence of this observation and
  Theorem~\ref{thm:uniform-prob}, it suffices to check that
  $\MDPProb{\strat{1}}_{\StochG^{\uniformstrat{2}},v}(\Diamond T)=1$ for every
  strategy $\strat{1}$ iff
  $v \in V\setminus \Epre^*(V \setminus \Apre^*(T))$ $(\ddag)$.

  By Lemma~10.110~\cite{BaierK08}, $v\in V \setminus \Apre^*(T))$ iff
  $\exists \strat{1}\colon \MDPProb{\strat{1}}_{\StochG^{\uniformstrat{2}},v}(\Diamond T) = 0$.
  Therefore, $v\in V \setminus \Apre^*(T))$ iff
  $\exists \strat{1}\colon \MDPProb{\strat{1}}_{\StochG^{\uniformstrat{2}},v}(\Box \neg T) = 1$.
  Taking into account that all terminal states are absorbing, from Lemma~10.111~\cite{BaierK08}, $v\in \Epre^*(V \setminus \Apre^*(T)))$ iff
  $\exists \strat{1}\colon \MDPProb{\strat{1}}_{\StochG^{\uniformstrat{2}},v}(\Diamond T) < 1$
  from which $(\ddag)$ follows.
  \qed
\end{myproof}

\begin{myproof}[of Theorem~\ref{th:semmimarkov2}]
  Consider the event $\Diamond^{k} v' = \{ \omega \in
  \GamePaths_\StochG \mid \omega_k = v'\}$, for $k\geq 0$. That is,
  the set of runs in which $v'$ is reached in exactly $k$ steps.
  We define $\starredstrat{1}$ as follows:
  \[
  \starredstrat{1}(\hat{\omega}v')(v'') =  \Prob{\strat{1}}{\strat{2}}_{\StochG,v}(\Diamond^{k+1} v'' \mid \Diamond^k v') 
  \]
  for every $\hat{\omega}$ such that
  $\Prob{\strat{1}}{\strat{2}}_{\StochG,v}(\hat{\omega}v') > 0$ and
  $|\hat{\omega}v'| = k$.  For $\hat{\omega}$ with
  $\Prob{\strat{1}}{\strat{2}}_{\StochG,v}(\hat{\omega}v') = 0$ we
  define $\starredstrat{1}(\hat{\omega}v')$ to be any arbitrary
  distribution.

  We prove that this strategy satisfies the conditions of the theorem.
  First, we prove that
  \begin{equation}\label{equ:diamondk}
    \Prob{\strat{1}}{\strat{2}}_{\StochG,v}(\Diamond^{k} v') =
    \Prob{\starredstrat{1}}{\strat{2}}_{\StochG,v}(\Diamond^{k} v').
  \end{equation}
  Notice that for $k=0$ we have that,
  \[\Prob{\strat{1}}{\strat{2}}_{\StochG,v}(\Diamond^{0} v') =
  \Prob{\strat{1}}{\strat{2}}_{\StochG,v}(v') = \Delta_v(v') =
  \Prob{\starredstrat{1}}{\strat{2}}_{\StochG,v}(v')=
  \Prob{\starredstrat{1}}{\strat{2}}_{\StochG,v}(\Diamond^{0} v').\]
  For $k + 1 > 0$, we note that
  \[
  \Prob{\strat{1}}{\strat{2}}_{\StochG,v}(\Diamond^{k+1} v') = \sum_{\hat{\omega} \in V^{k+1}} \Prob{\strat{1}}{\strat{2}}_{\StochG,v}(\hat{\omega}v')  = \sum_{v'' \in \pre(v')}\sum_{\hat{\omega} \in V^k} \Prob{\strat{1}}{\strat{2}}_{\StochG,v}(\hat{\omega}v''v'),
  \] 
  and hence, for this case, it sufficies to show that, for all
  $v''\in\pre(v)$,
  \[
    \sum_{\hat{\omega} \in V^k} \Prob{\strat{1}}{\strat{2}}_{\StochG,v}(\hat{\omega}v''v') = \sum_{\hat{\omega} \in V^k} \Prob{\starredstrat{1}}{\strat{2}}_{\StochG,v}(\hat{\omega}v''v').
  \]

  The proof uses induction on $k$ doing case analysis. We observe that
  the base case follows similarly to the inductive case and we
  properly point it out when it corresponds.
  If $v'' \in V_2$, we proceed as follows.
  \begin{align}	
    \label{th:semmimarkov2-eq1-l1}
    \sum_{\hat{\omega} \in V^k}\Prob{\strat{1}}{\strat{2}}_{\StochG,v}(\hat{\omega}v''v')
    & = \sum_{\hat{\omega} \in V^k}\Prob{\strat{1}}{\strat{2}}_{\StochG,v}(\hat{\omega}v'') \strat{2}(\hat{\omega}v'')(v')\\
    \label{th:semmimarkov2-eq1-l2}
    & = \alpha \sum_{\hat{\omega} \in V^k}\Prob{\strat{1}}{\strat{2}}_{\StochG,v}(\hat{\omega}v'') \\
    \label{th:semmimarkov2-eq1-l3}
    & = \alpha \sum_{\hat{\omega} \in V^k}\Prob{\starredstrat{1}}{\strat{2}}_{\StochG,v}(\hat{\omega}v'') \\
    \label{th:semmimarkov2-eq1-l4}
    & = \sum_{\hat{\omega} \in V^k}\Prob{\starredstrat{1}}{\strat{2}}_{\StochG,v}(\hat{\omega}v''v')
  \end{align}
  (\ref{th:semmimarkov2-eq1-l1}) follows by definition.
  For (\ref{th:semmimarkov2-eq1-l2}), we define $\alpha =
  \strat{2}(\hat{\omega}v'')(v')$ if $\hat{\omega}$ starts at $v$.  By
  recalling that the semi-Markov nature of $\strat{2}$ defines the
  same $\alpha$ for any $\hat{\omega}\in V^k$ starting at $v$, and that
  $\Prob{\strat{1}}{\strat{2}}_{\StochG,v}(\hat{\omega}v'') = 0$ if
  $\hat{\omega}$ does not start at $v$, $\alpha$ results to be common
  factor.
  (\ref{th:semmimarkov2-eq1-l3}) follows either by induction
  hypothesis, or, in the base case, by observing that
  $\Prob{\strat{1}}{\strat{2}}_{\StochG,v}(v'') =
  \Delta_v(v'')=\Prob{\starredstrat{1}}{\strat{2}}_{\StochG,v}(v'')$.
  (\ref{th:semmimarkov2-eq1-l4}) resolves just like
  (\ref{th:semmimarkov2-eq1-l2}).

  For the case in which $v'' \in V_1$, we proceed as follows.
  \begin{align}	
    \sum_{\hat{\omega} \in V^k}\Prob{\strat{1}}{\strat{2}}_{\StochG,v}(\hat{\omega}v''v') \hspace{-6em} & \notag\\
    \label{th:semmimarkov2-eq2-l1}
    & = \sum_{\hat{\omega} \in V^k}\Prob{\strat{1}}{\strat{2}}_{\StochG,v}(\hat{\omega}v'') \strat{1}(\hat{\omega}v'')(v')\\
    \label{th:semmimarkov2-eq2-l2}
    & = \left( \sum_{\hat{\omega} \in V^k}\Prob{\strat{1}}{\strat{2}}_{\StochG,v}(\hat{\omega}v'') \right) \frac{\displaystyle\sum_{\hat{\omega} \in V^k}\Prob{\strat{1}}{\strat{2}}_{\StochG,v}(\hat{\omega}v'') \strat{1}(\hat{\omega}v'')(v')}{\displaystyle\sum_{\hat{\omega} \in V^k}\Prob{\strat{1}}{\strat{2}}_{\StochG,v}(\hat{\omega}v'')} \\
    \label{th:semmimarkov2-eq2-l3}
    & = \sum_{\hat{\omega} \in V^k}\Prob{\strat{1}}{\strat{2}}_{\StochG,v}(\hat{\omega}v'') \Prob{\strat{1}}{\strat{2}}_{\StochG,v}(\Diamond^{k+1} v' \mid \Diamond^k v'') \\
    \label{th:semmimarkov2-eq2-l4}
    & = \sum_{\hat{\omega} \in V^k}\Prob{\starredstrat{1}}{\strat{2}}_{\StochG,v}(\hat{\omega}v'') \starredstrat{1}(\hat{\omega}v'')(v)\\
    \label{th:semmimarkov2-eq2-l5}
    & = \sum_{\hat{\omega} \in V^k}\Prob{\starredstrat{1}}{\strat{2}}_{\StochG,v}(\hat{\omega}v''v')
  \end{align}
  (\ref{th:semmimarkov2-eq2-l1}) follows by definition.
  (\ref{th:semmimarkov2-eq2-l2}) is obtained by multiplying and
  dividing the term by
  $\sum_{\hat{\omega}v'' \in V^k}\Prob{\strat{1}}{\strat{2}}_{\StochG,v}(\hat{\omega}v'')$.
  (\ref{th:semmimarkov2-eq2-l3}) follows from the definition of
  conditional probability by noting that
  $\Prob{\strat{1}}{\strat{2}}_{\StochG,v}({\Diamond^{k+1} v'} \cap {\Diamond^k v''}) =
  \sum_{\hat{\omega} \in V^k}\Prob{\strat{1}}{\strat{2}}_{\StochG,v}(\hat{\omega}v'') \strat{1}(\hat{\omega}v'')(v')$
  and 
  $\Prob{\strat{1}}{\strat{2}}_{\StochG,v}(\Diamond^{k} v'') =
  \sum_{\hat{\omega} \in V^k}\Prob{\strat{1}}{\strat{2}}_{\StochG,v}(\hat{\omega}v'')$.
  Finally, (\ref{th:semmimarkov2-eq2-l4}) follows by the definition of
  $\starredstrat{1}$ and (\ref{th:semmimarkov2-eq2-l5}) by the
  definition of $\Prob{\starredstrat{1}}{\strat{2}}_{\StochG,v}$.

  The proof for case $v'' \in V_\Probabilistic$ follows just like for
  the case first case, with the difference that instead of considering
  $\strat{2}(\hat{\omega}v'')(v')$ we need to consider
  $\delta(v'',v')$.

  Now,  we have:
  \begin{align*}
  \Expect{\strat{1}}{\strat{2}}_{\StochG,v}[\Rewards]   &  = \sum^\infty_{N=0} \sum_{\hat{\omega} \in V^{N+1}} \Prob{\strat{1}}{\strat{2}}_{\StochG,v}(\hat{\omega})\reward(\hat{\omega}_N)\\
  & = \sum^\infty_{N=0} \sum_{v \in V} \Prob{\strat{1}}{\strat{2}}_{\StochG,v}(\Diamond^N v)\reward(v) \\
  & =  \sum^\infty_{N=0} \sum_{v \in V} \Prob{\starredstrat{1}}{\strat{2}}_{\StochG,v}(\Diamond^N v)\reward(v) \\
  & = \Expect{\starredstrat{1}}{\strat{2}}_{\StochG,v}[\Rewards] 
  \end{align*}
  \qed
\end{myproof}

\begin{myproof}[of Lemma~\ref{lemma:bound-prob-stationary-strats}]
  Since $\strat{1}$ and $\strat{2}$ are memoryless, we can fix these
  strategies and obtain an absorbing finite Markov chain
  $\StochG^{\strat{1},\strat{2}} = (V, (\emptyset,\emptyset, V_1 \cup V_2 \cup V_P), \delta^{\strat{1}, \strat{2}})$,
  then the result follows from basic properties of Markov chains. For
  the sake of completeness, we provide a proof.
  
  Let $|V|$ be the number of states of the MC
  $\StochG^{\strat{1},\strat{2}}$ and $T$ its set of terminal
  vertices.
  Because $\Prob{\strat{1}}{\strat{2}}_{\StochG,v}(\Diamond T) = 1$
  and the MC is finite, from any vertex, we have a path of length $k$
  with $k \leq|V|$ to a terminal vertex.
  Furthermore, let
  $\lambda = \min \{ \delta^{\strat{1},\strat{2}}(v,v') > 0 \mid v,v' \in V \}$,
  which is well-defined since $\StochG^{\strat{1},\strat{2}}$ has a finite
  number of states. Thus for any $N\geq 0$ we have:
  \[
  \sum_{\hat{\omega} \in (V \setminus T)^{N+1}} \Prob{\strat{1}}{\strat{2}}_{\StochG,v}(\hat{\omega}) \leq (1 - \lambda^{|V|})^{\left\lfloor{\frac{N}{|V|}}\right\rfloor}
  \]
  Thus, if $\lambda<1$,
  \begin{align*}
    \sum^\infty_{N=0}\sum_{\hat{\omega} \in (V \setminus T)^{N+1}} \Prob{\strat{1}}{\strat{2}}_{\StochG,v}(\hat{\omega}) & \ \leq \ \sum^{\infty}_{N=0}(1 - \lambda^{|V|})^{\left\lfloor{\frac{N}{|V|}}\right\rfloor} \\
    & \ = \ {|V|}\sum^{\infty}_{i=0}(1 - \lambda^{|V|})^i \ = \ \frac{|V|}{1 - \lambda^{|V|}}.
  \end{align*}
  If $\lambda = 1$, it is easy to check that
  $\sum^\infty_{N=0}\sum_{\hat{\omega} \in (V \setminus T)^{N+1}} \Prob{\strat{1}}{\strat{2}}_{\StochG,v}(\hat{\omega}) \leq {|V|}$.
  \qed
\end{myproof}

\begin{myproof}[of Theorem~\ref{th:games-are-bounded}]
  Rather than prooducing the approriate reductions to apply Theorem
  4.2.12 in \cite[p.~174]{FilarV96}, we produce a direct proof based
  on the ideas of the proof of that Theorem.
  
  Fix a fair strategy $\strat{2} \in \MemorylessFairStrats{2}$ and
  consider the MDP
  $\StochG^\strat{2} = (V, (V_1, \emptyset, V_2 \cup V_\Probabilistic), \delta^{\strat{2}})$.
  Define the following functional:
  \[
  \Gamma'(f)(v) =
  \begin{cases}
    \sum_{v' \in \post(v)} \delta^\strat{2}(v,v')  f(v') & \text{ if } v \in V_2 \cup V_\Probabilistic  \\
    \max \{1  + f(v') \mid v' \in \post(v) \} & \text{ if } v \in  V_1 \setminus T \\
    0 & \text{ if } v \in T.
  \end{cases}
  \]
  It is well-known that, since $\StochG^{\strat{2}}$ is stopping, this
  functional is a contraction mapping and then it has a unique
  fixpoint (see, e.g., proof of Theorem 4.2.2 in \cite{FilarV96}),
  which gives us the optimal value for Player~1 in MDP
  $\StochG^{\strat{2}}$.  Let $\nu$ denote the fixpoint of $\Gamma'$.
  Since $\nu(v_n)$ is the optimal value for a given vertex $v_{n}$,
  we have that
  \begin{equation}\label{eq:th:games-are-bounded-eq1}
    \nu(v_n) \geq 1 + \sum_{v_{n+1} \in \post(v_n)} \delta^{\strat{1}^n, \strat{2}}(v_n, v_{n+1}) \nu(v_{n+1})
  \end{equation}
  for any memoryless strategy $\strat{1}^n$ (the value obtained by the
  fixpoint is optimal among memoryless strategies, thus it cannot be
  improved by $\strat{1}^n$).
  Take any sequence $\hat{\omega} = v_0 \dots v_n$ with
  $\hat{\omega}_n = v_n$ and multiply
  (\ref{eq:th:games-are-bounded-eq1}) by
  $\prod^{n-1}_{i=0} \delta^{\strat{1}^i, \strat{2}}(\hat{\omega}_i, \hat{\omega}_{i+1})$
  on both sides.  This yields
  \begin{align}
    \prod^{n-1}_{i=0} \delta^{\strat{1}^i, \strat{2}}(\hat{\omega}_i, \hat{\omega}_{i+1}) \nu(\hat{\omega}_n)  \hspace{-6.5em}& \notag\\
    \geq \ & \prod^{n-1}_{i=0} \delta^{\strat{1}^i, \strat{2}}(\hat{\omega}_i, \hat{\omega}_{i+1}) \label{eq:th:games-are-bounded-eq2} \\
    & + \ \prod^{n-1}_{i=0} \delta^{\strat{1}^i, \strat{2}}(\hat{\omega}_i, \hat{\omega}_{i+1})\sum_{v_{n+1} \in \post(v_n)} \delta^{\strat{1}^n, \strat{2}}(v_n, v_{n+1}) \nu(v_{n+1}) \notag
  \end{align}
  In the equation above, we assume
  $\prod^{-1}_{i=0} \delta^{\strat{1}^i, \strat{2}}(\hat{\omega}_i, \hat{\omega}_{i+1}) = 1$
  for the case in which $n=0$.
  
  Considering all sequences of length $n+1$ in $v(V \setminus T)^n$,
  from (\ref{eq:th:games-are-bounded-eq2}), and rewriting the last
  line of the inequality, we get
  \begin{align}
    \sum_{\hat{\omega} \in v(V \setminus T)^n}\prod^{n-1}_{i=0} \delta^{\strat{1}^i, \strat{2}}(\hat{\omega}_i, \hat{\omega}_{i+1}) \nu(\hat{\omega}_n)  \hspace{-5.5em}& \notag\\
    \geq \ & \sum_{\hat{\omega} \in v(V \setminus T)^n}\prod^{n-1}_{i=0} \delta^{\strat{1}^i, \strat{2}}(\hat{\omega}_i, \hat{\omega}_{i+1}) \label{eq:th:games-are-bounded-eq3}\\
    & + \ \sum_{\hat{\omega} \in v(V \setminus T)^{n+1}}(\prod^{n}_{i=0} \delta^{\strat{1}^i, \strat{2}}(\hat{\omega}_i, \hat{\omega}_{i+1}) \nu(\hat{\omega}_{n+1})) \notag
  \end{align}
  Summing all the possible sequences up to $N$ we obtain
  \begin{align}
    \sum^{N}_{n=0}\sum_{\hat{\omega} \in v(V \setminus T)^n}\prod^{n-1}_{i=0} \delta^{\strat{1}^i, \strat{2}}(\hat{\omega}_i, \hat{\omega}_{i+1}) \nu(\hat{\omega}_n)   \hspace{-6.5em}& \notag\\
    \geq \ & \sum^{N}_{n=0} \sum_{\hat{\omega} \in v(V \setminus T)^n}\prod^{n-1}_{i=0} \delta^{\strat{1}^i, \strat{2}}(\hat{\omega}_i, \hat{\omega}_{i+1}) \label{eq:th:games-are-bounded-eq4} \\
    & + \ \sum^{N}_{n=0} \sum_{\hat{\omega} \in v(V \setminus T)^{n+1}}(\prod^{n}_{i=0} \delta^{\strat{1}^i, \strat{2}}(\hat{\omega}_i, \hat{\omega}_{i+1}) \nu(\hat{\omega}_{n+1})) \notag
  \end{align}
  (\ref{eq:th:games-are-bounded-eq4}) can be equivalently rewritten
  as
  \begin{align}
    \sum^{N}_{n=0}\sum_{\hat{\omega} \in v(V \setminus T)^n}\prod^{n-1}_{i=0} \delta^{\strat{1}^i, \strat{2}}(\hat{\omega}_i, \hat{\omega}_{i+1}) \nu(\hat{\omega}_n)  \hspace{-6.5em}& \notag\\
    \geq \ & \sum^{N}_{n=0} \sum_{\hat{\omega} \in v(V \setminus T)^n}\prod^{n-1}_{i=0} \delta^{\strat{1}^i, \strat{2}}(\hat{\omega}_i, \hat{\omega}_{i+1}) \label{eq:th:games-are-bounded-eq5}\\
    & + \ \sum^{N+1}_{n=1} \sum_{\hat{\omega} \in v(V \setminus T)^{n}}(\prod^{n-1}_{i=0} \delta^{\strat{1}^i, \strat{2}}(\hat{\omega}_i, \hat{\omega}_{i+1}) \nu(\hat{\omega}_{n})) \notag
  \end{align}
  By subtracting
  $\sum^{N+1}_{n=1} \sum_{\hat{\omega} \in v(V \setminus T)^{n}}(\prod^{n-1}_{i=0} \delta^{\strat{1}^i, \strat{2}}(\hat{\omega}_i, \hat{\omega}_{i+1}) \nu(v_{n}))$
  from both sides of the inequality
  (\ref{eq:th:games-are-bounded-eq5}), we get
  \begin{align}
    \label{eq:th:games-are-bounded-eq6}
    \nu(v)  - \sum_{\hat{\omega} \in v(V \setminus T)^{N+1}}\prod^{N}_{i=0} \delta^{\strat{1}^i, \strat{2}}(\hat{\omega}_i,\hat{\omega}_{i+1}) \nu(\hat{\omega}_{N+1})  \hspace{-6.5em}& \\
    \geq \ & \sum^{N}_{n=0} \sum_{\hat{\omega} \in v(V \setminus T)^n}\prod^{n-1}_{i=0} \delta^{\strat{1}^i, \strat{2}}(\hat{\omega}_i, \hat{\omega}_{i+1}) \notag
  \end{align}
  Applying limits to both sides of the inequality
  (\ref{eq:th:games-are-bounded-eq6}) we obtain
  \begin{align}
    \label{eq:th:games-are-bounded-eq7}
    \nu(v)  - \lim_{N \to \infty} \sum_{\hat{\omega} \in v(V \setminus T)^{N+1}}\prod^{N}_{i=0} \delta^{\strat{1}^i, \strat{2}}(\hat{\omega}_i,\hat{\omega}_{i+1}) \val(\hat{\omega}_{N+1})  \hspace{-8.5em}& \\
    \geq \ & \lim_{N \to \infty} \sum^{N}_{n=0} \sum_{\hat{\omega} \in v(V \setminus T)^n}\prod^{n-1}_{i=0} \delta^{\strat{1}^i, \strat{2}}(\hat{\omega}_i, \hat{\omega}_{i+1}) \notag
  \end{align}
  Now, by Lemma \ref{lemma:bound-prob-stationary-strats}, $\nu(v)$ is
  bounded for every $v$.  In addition,
  \[
  \lim_{N \to \infty}\sum_{\hat{\omega} \in v(V \setminus T)^{N+1}}\prod^{N}_{i=0} \delta^{\strat{1}^i, \strat{2}}(v_i,v_{i+1}) \nu(v_{N+1}) = 0,
  \]
  Thus, the right term of the inequality
  (\ref{eq:th:games-are-bounded-eq7}) is a finite number.  That means
  that for every sequence of memoryless strategies
  $\strat{1}^0,\strat{1}^1, \dots$ we get
  \[
  \lim_{N \to \infty} \sum^{N}_{n=0} \sum_{\hat{\omega} \in v(V \setminus T)^n}\prod^{n-1}_{i=0} \delta^{\strat{1}^i, \strat{2}}(\hat{\omega}_i, \hat{\omega}_{i+1}) < \infty,
  \]
  Noting that any semi-Markov strategy $\strat{1}'$ can be seen as a
  sequence of memoryless strategies, it follows that
  \[
  \lim_{N \to \infty} \sum^{N}_{n=0} \sum_{\hat{\omega} \in v(V \setminus T)^n} \MDPProb{\strat{1}'}_{\StochG^{\strat{2}}, v}(\hat{\omega}) < \infty,
  \]
  for any $\strat{1}'$ semi-Markovian.  Furthermore, since every
  strategy $\strat{1}$ has an equivalent semi-Markov strategy (see
  proof of Theorem \ref{th:semmimarkov2}), the result follows.
  \qed
\end{myproof}

\begin{myproof}[of Theorem~\ref{th:memoryless-strat-p2-bounded-expectation}]
  Take $M = \max \{r(v) \mid v \in V\}$.  This number is well defined
  since the number of vertices is finite.
  The expected value of the game under strategies
  $\strat{1} \in \Strategies{1}$ and $\strat{2} \in \FairStrats{2}$
  is given by
  \begin{align}	
    \label{th:memoryless-strat-p2-bounded-expectation-eq1-l0}
    \Expect{\strat{1}}{\strat{2}}_{\StochG,v}[\Rewards]
    & = \sum^{\infty}_{N=0} \sum_{\hat{\omega} \in V^{N+1}} \Prob{\strat{1}}{\strat{2}}_{\StochG,v}(\hat{\omega})\reward(\hat{\omega}_N) \\
    \label{th:memoryless-strat-p2-bounded-expectation-eq1-l1}	
    & = \sum^{\infty}_{N=0} \ \ \sum_{\hat{\omega} \in (V^{N+1} \cap (V \setminus T)^*T)} \Prob{\strat{1}}{\strat{2}}_{\StochG,v}(\hat{\omega})  \reward(\hat{\omega}_N) \ \ + \notag\\
    &\hspace*{4em} \sum_{\hat{\omega} \in (V^{N+1} \cap (V \setminus T)^*)} \Prob{\strat{1}}{\strat{2}}_{\StochG,v}(\hat{\omega}) \reward(\hat{\omega}_N)   \\
    \label{th:memoryless-strat-p2-bounded-expectation-eq1-l2}
    & =  \sum^{\infty}_{N=0}  (\sum_{\hat{\omega} \in V^{N+1} \cap (V \setminus T)^*} \Prob{\strat{1}}{\strat{2}}_{\StochG,v}(\hat{\omega}) \reward(\hat{\omega}_N) ) \\
    \label{th:memoryless-strat-p2-bounded-expectation-eq1-l3}
    & \leq M \sum^{\infty}_{N=0}  \sum_{\hat{\omega} \in V^{N+1} \cap (V \setminus T)^*} \Prob{\strat{1}}{\strat{2}}_{\StochG,v}(\hat{\omega})  \\
    \label{th:memoryless-strat-p2-bounded-expectation-eq1-l4}
    & < \infty 
  \end{align}
  (\ref{th:memoryless-strat-p2-bounded-expectation-eq1-l0}) is the
  definition of the expected value.
  The fact that $V^{N+1} \cap (V \setminus T)^*T$ and
  $V^{N+1} \cap (V \setminus T)^*$ are disjoint justifies
  (\ref{th:memoryless-strat-p2-bounded-expectation-eq1-l1}).
  (\ref{th:memoryless-strat-p2-bounded-expectation-eq1-l2}) follows
  from the assumption that $\reward(v) = 0$ for all $v \in T$ and
  (\ref{th:memoryless-strat-p2-bounded-expectation-eq1-l3}) from the
  definition of $M$.
  Finally, (\ref{th:memoryless-strat-p2-bounded-expectation-eq1-l4})
  is a consequence of Theorem \ref{th:games-are-bounded}.
  \qed
\end{myproof}

\begin{myproof}[of Theorem~\ref{th:infima-in-dmf}]
  By fixing strategy $\strat{1}$ on $\StochG$ we obtain the MDP
  $\StochG^{\strat{1}} = (V, (\emptyset,V_2,V_1\cup V_\Probabilistic),\delta^{\strat{1}})$
  where $\delta^{\strat{1}}(v,\cdot)=\strat{1}(v)$ if $v\in V_1$, and
  $\delta^{\strat{1}}(v,\cdot)=\delta(v,\cdot)$ if $v\in V_2\cup V_\Probabilistic$.
  Notice that the set $V_1$ of Player~1 vertices in $\StochG$ became
  part of the probabilistic vertices of $\StochG^{\strat{1}}$.  Thus,
  non-deterministic choices can only be present at vertices in $V_2$.
  By definition,
  $\inf_{\strat{2} \in \FairStrats{2}} \ExpectMDP{\strat{2}}_{\StochG^{\strat{1}},v}[\Rewards] =
  \inf_{\strat{2} \in \FairStrats{2}} \Expect{\strat{1}}{\strat{2}}_{\StochG,v}[\Rewards]$
  for all $v\in V$.

  Though the proof in general differs, the basic strategy is inspired
  by the proof of Lemma~10.102 in~\cite{BaierK08}: We first construct
  a reduced MDP $\StochG^{\strat{1}}_{\min}$ which preserves the
  optimizing values of $\StochG^{\strat{1}}$ at each vertex, and then
  use the structure of $\StochG^{\strat{1}}_{\min}$ to derive an
  optimal deterministic memoryless strategy.

  For all $v\in V$,
  let $x_v = \inf_{\strat{2} \in \FairStrats{2}} \Expect{\strat{1}}{\strat{2}}_{\StochG,v}[\Rewards]$ (by Corollary~\ref{coro:inf-for-strat2-is-bounded} this value is well-defined) and,
  for all $v\in V_2$, let
  \begin{equation}\label{th:infima-in-dmf-def-postmin}
 	\postmin(v) = \{{v'\in V} \mid {\delta^{\strat{1}}(v,v')=1} \land {x_v=\reward(v)+x_{v'}}\}.
  \end{equation}
  We define the MDP
  $\StochG^{\strat{1}}_{\min} = (V, (\emptyset,V_2,V_1\cup V_\Probabilistic),\delta^{\strat{1}}_{\min})$
  where $\delta^{\strat{1}}_{\min}(v,v')=\delta^{\strat{1}}(v,v')$ if $v\in V_1\cup
  V_\Probabilistic$, or $v\in V_2$ and $v'\in\postmin(v)$. Otherwise
  $\delta^{\strat{1}}_{\min}(v,v')=0$.
  Thus, $\StochG^{\strat{1}}_{\min}$ is the same as $\StochG^{\strat{1}}$ except that all
  transitions $\delta^{\strat{1}}(v,v')=1$ where $v\in V_2$ and
  $x_v < \reward(v)+x_{v'}$ have been removed.
  Notice that
  $\inf_{\strat{2} \in \FairStrats{2}} \ExpectMDP{\strat{2}}_{\StochG^{\strat{1}}_{\min},v}[\Rewards] =
  \inf_{\strat{2} \in \FairStrats{2}} \ExpectMDP{\strat{2}}_{\StochG^{\strat{1}},v}[\Rewards] =
  x_v$
  for all $v\in V$.

  Before continuing, we prove the following claim
  \begin{claim}
    $\StochG^{\strat{1}}_{\min}$ is stopping under fairness.
  \end{claim}
  \begin{proofofclaim}
    We proceed by contradiction.  Suppose that there is a fair
    strategy $\strat{2}'$ such that, for some node $v \in V$ we have
    $\MDPProb{\strat{2}'}_{\StochG_{\min}^{\strat{1}},v}(\Diamond T) < 1$.
    Hence,  $\MDPProb{\strat{2}'}_{\StochG_{\min}^{\strat{1}},v}(\Box \neg T) > 0$.
    This implies that there is an end component $\EC{C} = (V',
    \delta')$ in the MDP $\StochG_{\min}^{\strat{1}}$ such that
    $T \cap V' = \emptyset$ and
    $\MDPProb{\strat{2}'}_{\StochG_{\min}^{\strat{1}},v}(\Diamond V') > 0$
    (Theorem 10.133 \cite[p.~889]{BaierK08}).
    Since $\StochG_{\min}^{\strat{1}}$ is a sub-MDP of
    $\StochG^{\strat{1}}$, $\EC{C}$ is also an end component of
    $\StochG^{\strat{1}}$.
    Also, note that $\EC{C}$ cannot be a maximal end component in
    $\StochG^{\strat{1}}$, otherwise we have found a fair strategy
    reaching a maximal end component and not reaching $T$, thus
    violating the assumption that $\StochG$ is stopping under
    fairness.
    
    Let $m = \min \{ x_{v'} \mid v' \in V' \}$ and
    $\hat{v} \in \argmin  \{ x_{v'}  \mid v' \in V' \}$.
    If $\hat{v}\in V_2$, by definition of $\postmin(v)$,
    $x_{\hat{v}} = \reward(\hat{v}) + x_{v'}$, for all
    $v' \in \postmin(v)$, and hence, necesaryly $\reward(\hat{v})=0$
    and $x_{\hat{v}} = x_{v'}$ since all values are non-negative.
    The same holds in case $\hat{v}$ is a probabilistic vertex,
    because $x_{\hat{v}}$ depends on the convex combination of the
    values of its successors.
    Thus, inductively, $x_{\hat{v}} = x_{v'}$ and $\reward(v')=0$ for
    all $v'\in V'$.

    Now, let
    $M =  \min \{ \inf_{\strat{2} \in \FairStrats{2}}  \ExpectMDP{\strat{2}}_{\StochG^{\strat{1}},v'}[\Rewards] \mid v' \in (\post(V') \setminus V')\}$
    and 
    $F=\{v \in V \mid v \in (\post(V') \setminus V')\}$.
    Note that, by definition of $\StochG_{min}^{\strat{1}}$, we have that $M > m$.
    Set $\epsilon = \frac{M -m}{2}$ and consider an $\epsilon$-optimal
    strategy $\hat{\strat{2}}$ for vertex $\hat{v}$, i.e.,
    $\ExpectMDP{\hat{\strat{2}}}_{\StochG^{\strat{1}},\hat{v}}[\Rewards] \leq x_{\hat{v}} + \epsilon$.
    
    Note that, since $\StochG^{\strat{1}}$ is stopping under fairness,
    when using strategy $\hat{\strat{2}}$, any play almost surely
    leaves $\EC{C}$ (recall that $V' \cap T = \emptyset$) and hence
    $\MDPProb{\hat{\strat{2}}}_{\StochG^{\strat{1}},\hat{v}} (\Diamond F) = 1$.
    Since all the rewards of the vertices in $\EC{C}$ are $0$ (as
    proven above) we can show that
    $\ExpectMDP{\hat{\strat{2}}}_{\StochG^{\strat{1}},\hat{v}}[\Rewards]
    \geq M$ as follows:
    \begin{align}	
      \ExpectMDP{\hat{\strat{2}}}_{\StochG^{\strat{1}},\hat{v}}[\Rewards] \hspace{-4em} & \notag\\
      \label{th:infima-in-dmf-eq1-l1}
      &= \sum_{\hat{\omega} \in (V \setminus T)^*T} \MDPProb{\hat{\strat{2}}}_{\StochG^{\strat{1}},\hat{v}}(\hat{\omega})\ \Rewards(\hat{\omega}) \\
      \label{th:infima-in-dmf-eq1-l2}
      &=\sum_{\hat{\omega} \in \hat{v}{V'}^*F(V \setminus T)^*T} \MDPProb{\hat{\strat{2}}}_{\StochG^{\strat{1}},\hat{v}}(\hat{\omega})\ \Rewards(\hat{\omega}) \\
      \label{th:infima-in-dmf-l3}
      &= \sum_{v' \in F} \sum_{\hat{\omega} \in \hat{v}{V'}^*v'} \MDPProb{\hat{\strat{2}}}_{\StochG^{\strat{1}},\hat{v}}(\hat{\omega}) \sum_{\hat{\omega}' \in v' (V \setminus T)^*T}  \MDPProb{\hat{\strat{2}}}_{\StochG^{\strat{1}},v'}(\hat{\omega}')\ \Rewards(\hat{\omega}')\\
      \label{th:infima-in-dmf-l4}
      &\geq \sum_{v' \in F} \sum_{\hat{\omega} \in \hat{v}{V'}^*v'} \MDPProb{\hat{\strat{2}}}_{\StochG^{\strat{1}},\hat{v}}(\hat{\omega})\ x_{v'} \\
      \label{th:infima-in-dmf-l5}
      &\geq \sum_{v' \in F} \sum_{\hat{\omega} \in \hat{v}{V'}^*v'} \MDPProb{\hat{\strat{2}}}_{\StochG^{\strat{1}},\hat{v}}(\hat{\omega})\ M \\
      \label{th:infima-in-dmf-l6}
      &= M  \sum_{v' \in F} \sum_{\hat{\omega} \in \hat{v}{V'}^*v'} \MDPProb{\hat{\strat{2}}}_{\StochG^{\strat{1}},\hat{v}}(\hat{\omega}) \\
      \label{th:infima-in-dmf-l7}
      &= M
    \end{align}
    (\ref{th:infima-in-dmf-eq1-l1}) is the definition of expectation.
    (\ref{th:infima-in-dmf-eq1-l2}) follows by the observation that
    any path $\hat{\omega} \in (V \setminus T)^*T$ needs to start in
    $\hat{v}\in V'$ and needs to pass through the frontier $F$ to
    leave $\EC{C}$ and reach some terminal state in $T$.
    (\ref{th:infima-in-dmf-l3}) is obtained by taking into account
    that $\reward(v) = 0$ for any $v \in V'$.
    (\ref{th:infima-in-dmf-l4}) follows from the fact that 
    $x_{v'} \leq \sum_{\hat{\omega}' \in v' (V \setminus T)^*T}  \MDPProb{\hat{\strat{2}}}_{\StochG^{\strat{1}},v'}(\hat{\omega}')\ \Rewards(\hat{\omega}')$.
    (\ref{th:infima-in-dmf-l5}) follows by the definition of $M$ and
    (\ref{th:infima-in-dmf-l6}) by factorizing $M$.
    Finally, (\ref{th:infima-in-dmf-l7}) follows from the fact 
    $\sum_{v' \in F} \sum_{\hat{\omega} \in \hat{v}{V'}^*v'} \MDPProb{\hat{\strat{2}}}_{\StochG^{\strat{1}},\hat{v}}(\hat{\omega}) =
    \MDPProb{\hat{\strat{2}}}_{\StochG^{\strat{1}},\hat{v}} (\Diamond F) = 1$.
    Thus we have that
    $x_{\hat{v}} + \epsilon \geq \ExpectMDP{\hat{\strat{2}}}_{\StochG^{\strat{1}},\hat{v}}[\Rewards] \geq M$.
    Since also $M >  x_{\hat{v}} + \epsilon$, we reach  a contradiction.
    Hence $\StochG^{\strat{1}}_{\min}$ must be stopping under fairness.
    \hfill\emph{(End of claim)}\qed
  \end{proofofclaim}

  For every $v\in V$, let $\disttoT{v}$ be the length of the shortest
  path fragment to some terminal vertex in $T$ in the MDP
  $\StochG^{\strat{1}}_{\min}$.   In
  particular, $\disttoT{v}=0$ for all $v\in T$.  Notice that $\disttoT{v}$ is defined for all
  $v\in V$ because $\StochG^{\strat{1}}_{\min}$ inherits from $\StochG$ the
  property of being stopping under fair strategies (and hence almost
  surely reaches $T$ for any fair strategy in $\FairStrats{2}$).
  
  Now,  for every $v\in V_2$
  such that $\disttoT{v} \geq 1$, define $\starredstrat{2}(v)(v')=1$
  for some $v'$ such that $\delta^{\strat{1}}_{\min}(v,v')=1$ and
  $\disttoT{v}=\disttoT{v'}+1$; and $\starredstrat{2}(v)(v'')=0$ for $v''\neq v'$.  Notice that such $v'$ always exists,
  and that $\starredstrat{2}$ is a memoryless and deterministic
  strategy.

  $\starredstrat{2}$ induces the (finite) Markov chain
  $\StochG^{\strat{1},\starredstrat{2}} = (V, (\emptyset,\emptyset,V_1\cup V_2\cup V_\Probabilistic),\delta^{\strat{1},\starredstrat{2}})$,
  where
  \[\delta^{\strat{1},\starredstrat{2}}(v,v') =
    \begin{cases}
      \delta(v,v')     & \text{ if } v\in V_\Probabilistic \\
      \strat{1}(v)(v') & \text{ if } v\in V_1 \\
      \starredstrat{2}(v)(v') & \text{ if } v\in V_2.
    \end{cases}
  \]
  Taking into account the definition of $\delta^{\strat{1},\starredstrat{2}}$,
  the expected reward values
  $\ExpectMDP{}_{\StochG^{\strat{1},\starredstrat{2}},v}[\Rewards]$ are obtained
  by the unique solution of the following linear equation system:
  \begin{align*}
    y_v = {} & 0 && \text{ if } v\in T \\
    y_v = {} &\textstyle \reward(v) + \sum_{v\in V} \delta(v,v') \cdot y_{v'} && \text{ if } v\in V_\Probabilistic \setminus T\\
    y_v = {} &\textstyle \reward(v) + \sum_{v\in V} \strat{1}(v)(v') \cdot y_{v'}  && \text{ if } v\in V_1 \setminus T\\
    y_v = {} & \reward(v) +  y_{v'} && \text{ if } v\in V_2 \setminus T \text{ and } \starredstrat{2}(v)(v') = 1
  \end{align*}
  To see that this equation system has a unique solution, consider its
  matrix form $y = Ay+ r$, wherein $A$ is a matrix defined by
  $A_{i,j} = \delta^{\pi_1,\pi^*_2}(v_i,v_j)$ if $v_i \notin T$ and
  $A_{i,j} = 0$ if $v_i \in T$, and $r$ is the reward vector where
  $r_i = \reward(v_i)$.
  $A$ is a square substochastic matrix
  \cite{DBLP:journals/moc/Azimzadeh19} in which there is a path from
  any vertex to a terminal vertex.
  Thus, by Corollary 2.6 in \cite{DBLP:journals/moc/Azimzadeh19},
  $(I-A)^{-1}$ exists and hence $(I-A)^{-1} r$ is the unique solution
  to the equation system given above.

  Since $\starredstrat{2}(v)(v') = 1$ implies that $v'\in\postmin(v)$,
  it follows that $(x_v)_{v\in V}$ also solves the equation system
  above.  By uniqueness of the solution we get
  \[\Expect{\strat{1}}{\starredstrat{2}}_{\StochG,v}[\Rewards] = 
    \ExpectMDP{}_{\StochG^{\strat{1},\starredstrat{2}},v}[\Rewards] =
    y_v = x_v =
    \inf_{\strat{2} \in \FairStrats{2}} \Expect{\strat{1}}{\strat{2}}_{\StochG,v}[\Rewards].
  \]

  In the last step of the proof we show $\starredstrat{2}$ is fair.
  By contradiction, suppose this is not the case.  Hence,
  $\Prob{\strat{1}}{\starredstrat{2}}_{\StochG,v}(\FP^2) < 1$ for some
  $v\in V$.  Thus, there must exist vertices $\hat{v}\in V_2$
  and $\hat{v}'\in\post(\hat{v})$ such that 
  $\Prob{\strat{1}}{\starredstrat{2}}_{\StochG,v}({\Box\Diamond\hat{v}}\land{\Diamond\Box\neg\hat{v}'})>0$,
  where
  ${\Box\Diamond\hat{v}}\land{\Diamond\Box\neg\hat{v}'} =\{ {\omega\in\GamePaths_{\StochG}} \mid {\hat{v} \in \inf(\omega)} \land {\hat{v}' \notin \inf(\omega)} \}$.

  By Theorem 10.56 in~\cite{BaierK08}, there is a BSCC
  $B$ in $\StochG^{\strat{1},\starredstrat{2}}$ satisfying $\hat{v}\in
  B$ and $\hat{v}'\notin B$, such that
  $\Prob{\strat{1}}{\starredstrat{2}}_{\StochG,v}({\Box\Diamond\hat{v}}\land{\Diamond\Box\neg\hat{v}'})
  = \Prob{\strat{1}}{\starredstrat{2}}_{\StochG,v}({\Diamond B})$.
  (A \emph{bottom strongly connected component (BSCC)} is a strongly
  connected component $B\subseteq V$ that cannot reach any vertex
  outside $B$, that is $\delta^{\strat{1},\starredstrat{2}}(v,B)=1$
  for all $v\in B$.)

  Take $v^*\in B$ with minimal distance in
  $\StochG^{\strat{1}}_{\min}$ to a terminal state in $T$, that is
  $\disttoT{v^*}\leq\disttoT{v}$ for all $v\in B$.  Since terminal
  vertices are absorbing, we have that $\hat{v}\notin T$, and hence
  $T\cap B=\emptyset$.  Therefore $\disttoT{v^*}>0$.
  If $v^*\in V_1 \cup V_\Probabilistic$, by definition of
  $\disttoT{v^*}$, $\delta^{\strat{1}}_{\min}(v^*,v')>0$
  for some $v'$ such that $\disttoT{v'}+1=\disttoT{v^*}$. Then
  $v'\notin B$, contradicting the fact that
  $\delta^{\strat{1}}_{\min}(v^*,B)=\delta^{\strat{1},\starredstrat{2}}(v^*,B)=1$.
  Hence, $v^*\notin V_1 \cup V_\Probabilistic$ $(\dagger)$.
  If $v^*\in V_2$, by definition,
  $\starredstrat{2}(v^*)(v')=1$ if and only if
  $\delta^{\strat{1}}_{\min}(v^*,v')=1$ and
  $\disttoT{v^*}=\disttoT{v'}+1$. Again $v'\notin B$, 
  which contradicts the fact that
  $\starredstrat{2}(v^*)(B)=\delta^{\strat{1},\starredstrat{2}}(v^*,B)=1$.
  Therefore $v^*\notin V_2$ $(\ddagger)$.
  
  By $(\dagger)$ and $(\ddagger)$, $v^*\notin V \supseteq B$
  contradicting our assumption that $v^*\in B$.  Therefore,
  $\starredstrat{2}$ is fair.
  \qed
\end{myproof}


To prove Theorem \ref{th:semimarkov-to-memoryless},  we need the following lemma.

\begin{lemma}\label{lemma:sum-of-nonterminal-is-zero}
  Let $\StochG$ a stochastic game that is stopping under fairness.
  Then for any $\strat{1} \in \Strategies{1}$,
  $\strat{2} \in \FairStrats{2}$, and $v \in V$ it holds:
  \[
  \lim_{N \to \infty} \sum_{\hat{\omega} \in (V \setminus T)^N} \Prob{\strat{1}}{\strat{2}}_{\StochG,v}(\hat{\omega}) = 0
  \]
  where $T$ is the set of terminal nodes of $\StochG$.
\end{lemma}
\begin{proof}
  First notice that the complement of $(V \setminus T)^N$ is
  $\Diamond^{\leq N} T = \cup^N_{i=0} \Diamond^i T$.
  Moreover, $\Diamond^{\leq N} T \subseteq \Diamond^{\leq N+1} T$, for all $N\geq 0$,
  and $\lim_{N \to \infty} \Diamond^{\leq N} T = \Diamond T$.
  Then,
  \begin{align*}
    \lim_{N \to \infty} \sum_{\hat{\omega} \in (V \setminus T)^N} \Prob{\strat{1}}{\strat{2}}_{\StochG,v}(\hat{\omega}) \
    & = \lim_{N \to \infty} \Prob{\strat{1}}{\strat{2}}_{\StochG,v}((V \setminus T)^N) \\
    & = \lim_{N \to \infty} 1 - \Prob{\strat{1}}{\strat{2}}_{\StochG,v}(\Diamond^{\leq N} T ) \\
    & = 1 - \lim_{N \to \infty} \Prob{\strat{1}}{\strat{2}}_{\StochG,v}(\Diamond^{\leq N} T ) \\
    & = 1 - \Prob{\strat{1}}{\strat{2}}_{\StochG,v}(\Diamond T ) \\
    & = 0
  \end{align*}
  The last equality is due to stopping under fairness, and all the
  previous ones, to standard measure theory results.
  \qed
\end{proof}

\begin{myproof}[of Theorem \ref{th:semimarkov-to-memoryless}]
  Consider a stochastic game $\StochG$ that is stopping under fairness, and a vertex $v$. First, note that any semi-Markov strategy $\strat{1}$ can be thought of as a  sequence of memoryless strategies: $\strat{1}^{0},\strat{1}^{1},\strat{1}^{2}, \dots$,  defined as follows:
  \begin{equation}\label{eq:def-memoryless-strat-for-semimarkov}
    \strat{1}^{n}(\hat{\omega}v)(v') = \strat{1}(v_0 \dots v_n)(v'), \text{ for any $\hat{\omega}$ and $v_0 \dots v_n$ with $v_n = v$}
  \end{equation}

  Now,  for any strategies $\strat{1} \in \SemiMarkovStrats{1}$ and $\strat{2} \in \MemorylessFairStrats{2}$, we have that:
  \[
  \Expect{\strat{1}}{\strat{2}}_{\StochG,v}[\Rewards] = 
  \sum^\infty_{n=0} \sum_{\hat{\omega} \in vV^n} \prod^{n-1}_{i=0} \delta^{\strat{1}^i,\strat{2}}(\hat{\omega}_i, \hat{\omega}_{i+1}) \reward(\hat{\omega}_n),
  \]
  where $\delta^{\strat{1}^i,\strat{2}}$ is the probability transition function of the Markov chain $\StochG^{\strat{1}^i,\strat{2}}$.  We assume that $\prod^{-1}_{i=0} \delta^{\strat{1}^i,\strat{2}}(\hat{\omega}_i, \hat{\omega}_{i+1}) = 1$.

  Fix a $\strat{2} \in \MemorylessFairStrats{2}$, and let $\starredstrat{1} \in \MemorylessStrats{1}$ be an optimal memoryless strategy for Player 1  for strategy $\strat{2}$.  We will prove that  no semi-Markov  strategy improves the value obtained by $\starredstrat{1}$.  Let $\strat{1}$ be an arbitrary and fixed semi-Markov strategy for Player 1.  First, note that, since $\starredstrat{1}$ is optimal for memoryless strategies, we have:
  \begin{equation}\label{eq:optimality-appendix}
    \reward(v_n) + \sum_{v_{n+1} \in \post(v_n)} \delta^{\strat{1}^n,\strat{2}}(v_n, v_{n+1}) \Expect{\starredstrat{1}}{\strat{2}}_{\StochG,v_{n+1}}[\Rewards] \leq \Expect{\starredstrat{1}}{\strat{2}}_{\StochG,v_n}[\Rewards]
  \end{equation}
  for any vertex $v_n$ and memoryless strategy $\strat{1}^n$ as defined in (\ref{eq:def-memoryless-strat-for-semimarkov}).  Now, let $\hat{\omega}$ be any sequence such that
  $\hat{\omega}_n = v_n$, we multiply both sides of (\ref{eq:optimality-appendix}) by $\prod^{n-1}_{i=0} \delta^{\strat{1}^i,\strat{2}}(\hat{\omega}_i,\hat{\omega}_{i+1})$,
  and get:
  \begin{align*}
    & \prod^{n-1}_{i=0}\delta^{\strat{1}^i, \strat{2}}_{\StochG,v}(\hat{\omega}_i, \hat{\omega}_{i+1}) \reward(v_n) \\
    & \quad + \ \prod^{n-1}_{i=0}\delta^{\strat{1}^i,\strat{2}}(\hat{\omega}_i, \hat{\omega}_{i+1}) \sum_{v_{n+1} \in \post(\hat{\omega}_n)}  \delta^{\strat{1}^n,\strat{2}}(\hat{\omega}_n, v_{n+1}) \Expect{\starredstrat{1}}{\strat{2}}_{\StochG,v_{n+1}}[\Rewards] \\
    & \hspace{18em} \leq 
    \prod^{n-1}_{i=0}\delta^{\strat{1}^i,\strat{2}}(\hat{\omega}_i,\hat{\omega}_{i+1}) \Expect{\starredstrat{1}}{\strat{2}}_{\StochG,v_n}[\Rewards]
  \end{align*}
  Thus,  considering all the sequences $\hat{\omega}$ of length $n+1$ we obtain:
  \begin{align*}
    & \sum_{\hat{\omega} \in vV^{n}} \prod^{n-1}_{i=0}\delta^{\strat{1}^i,\strat{2}}(\hat{\omega}_i,  \hat{\omega}_{i+1}) \reward(\omega_n) \\
    & \quad + \ \sum_{\hat{\omega} \in vV^{n+1}} \prod^{n}_{i=0}\delta^{\strat{1}^i,\strat{2}}(\hat{\omega}_i,\hat{\omega}_{i+1})  \Expect{\starredstrat{1}}{\strat{2}}_{\StochG,\hat{\omega}_{n+1}}[\Rewards]\\
    & \hspace{13em} \leq 
	\sum_{\hat{\omega} \in vV^{n}} \prod^{n-1}_{i=0}\delta^{\strat{1}^i,\strat{2}}(\hat{\omega}_i, \hat{\omega}_{i+1}) \Expect{\starredstrat{1}}{\strat{2}}_{\StochG,\hat{\omega}_n}[\Rewards]
  \end{align*}
  Summing up from $n=0$ to $N$:
  \begin{align*}
    & \sum^{N}_{n=0}\sum_{\hat{\omega} \in vV^{n}} \prod^{n-1}_{i=0}\delta^{\strat{1}^i,\strat{2}}(\hat{\omega}_i,\hat{\omega}_{i+1}) \reward(\hat{\omega}_n) \\
    & \quad + \ \sum^{N}_{n=0} \sum_{\hat{\omega} \in vV^{n+1}} \prod^{n}_{i=0}\delta^{\strat{1}^i,\strat{2}}(\hat{\omega}_i,\hat{\omega}_{i+1})  \Expect{\starredstrat{1}}{\strat{2}}_{\StochG,\hat{\omega}_{n+1}}[\Rewards] \\
    & \hspace{15em} \leq 
	\sum^{N}_{n=0}
	\sum_{\hat{\omega} \in vV^{n}} \prod^{n-1}_{i=0}\delta^{\strat{1}^i,\strat{2}}(\hat{\omega}_i, \hat{\omega}_{i+1}) \Expect{\starredstrat{1}}{\strat{2}}_{\StochG,\hat{\omega}_n}[\Rewards],
  \end{align*}
  %
  This is the same as:
  \begin{align*}
    & \sum^{N}_{n=0}\sum_{\hat{\omega} \in vV^{n}} \prod^{n-1}_{i=0}\delta^{\strat{1}^i,\strat{2}}(\hat{\omega}_i,\hat{\omega}_{i+1}) \reward(\hat{\omega}_n) \\
    & \quad + \
	\sum^{N+1}_{n=1} \sum_{\hat{\omega} \in vV^{n}} \prod^{n-1}_{i=0}\delta^{\strat{1}^i,\strat{2}}_{\StochG,v}(\hat{\omega}_i,\hat{\omega}_{i+1})  \Expect{\starredstrat{1}}{\strat{2}}_{\StochG,\hat{\omega}_{n}}[\Rewards] \\
    & \hspace{13em} \leq 
	 \sum^{N}_{n=0}
	\sum_{\hat{\omega} \in vV^{n}} \prod^{n-1}_{i=0}\delta^{\strat{1}^i,\strat{2}}(\hat{\omega}_i, \hat{\omega}_{i+1}) \Expect{\starredstrat{1}}{\strat{2}}_{\StochG,\hat{\omega}_n}[\Rewards]
  \end{align*}
  Thus, subtracting $\sum^{N+1}_{n=1} \sum_{\hat{\omega} \in vV^{n}} \prod^{n-1}_{i=0}\delta^{\strat{1}^i,\strat{2}}(\hat{\omega}_i,\hat{\omega}_{i+1})  \Expect{\starredstrat{1}}{\strat{2}}_{\StochG,\hat{\omega}_{n}}[\Rewards]$ from both sides we get:
  \begin{align*}
    & \sum^{N}_{n=0}\sum_{\hat{\omega} \in vV^{n}} \prod^{n-1}_{i=0}\delta^{\strat{1}^i,\strat{2}}(\hat{\omega}_i, \hat{\omega}_{i+1}) \reward(v_n)\\
    & \hspace{7em} \leq 
      \Expect{\starredstrat{1}}{\strat{2}}_{\StochG,v}[\Rewards] \ - \sum_{\hat{\omega} \in vV^{N+1}}\prod^{N}_{n=1}\delta^{\strat{1}^n,\strat{2}}(\hat{\omega}_i,\hat{\omega}_{i+1}) \Expect{\starredstrat{1}}{\strat{2}}_{\StochG, \hat{\omega}_{N+1}}[\Rewards]
  \end{align*}
  Applying limits on both sides, we obtain

  \begin{align}
    & \lim_{N \rightarrow \infty} \sum^{N}_{n=0}\sum_{\hat{\omega} \in vV^{n}} \prod^{n-1}_{i=0}\delta^{\strat{1}^i,\strat{2}}(\hat{\omega}_i, \hat{\omega}_{i+1}) \reward(v_n) \label{th:semimarkov-to-memoryless-eq0}\\
    & \hspace{5em} \leq 
      \Expect{\starredstrat{1}}{\strat{2}}_{\StochG,v}[\Rewards] \ - \lim_{N \rightarrow \infty} \sum_{\hat{\omega} \in vV^{N+1}}\prod^{N}_{n=1}\delta^{\strat{1}^n,\strat{2}}(\hat{\omega}_i,\hat{\omega}_{i+1}) \Expect{\starredstrat{1}}{\strat{2}}_{\StochG, \hat{\omega}_{N+1}}[\Rewards] \notag
  \end{align}

  Observe that
  \begin{align}
    \lim_{N \rightarrow \infty} \sum_{\hat{\omega} \in vV^{N+1}}\prod^{N}_{n=1}\delta^{\strat{1}^n,\strat{2}}(\hat{\omega}_i,\hat{\omega}_{i+1}) \Expect{\starredstrat{1}}{\strat{2}}_{\StochG, \hat{\omega}_{N+1}}[\Rewards] \hspace{-16em} & \notag \\
    = \ & \lim_{N \rightarrow \infty} \sum_{\hat{\omega} \in v(V\setminus T)^{N+1}}\prod^{N}_{n=1}\delta^{\strat{1}^n,\strat{2}}(\hat{\omega}_i,\hat{\omega}_{i+1}) \Expect{\starredstrat{1}}{\strat{2}}_{\StochG, \hat{\omega}_{N+1}}[\Rewards] \label{th:semimarkov-to-memoryless-lim-eq1}\\
    & \phantom{\lim_{N \rightarrow \infty}} + \
    \sum_{\hat{\omega} \in v\bigcup_{i=0}^N(V\setminus T)^iT^{N+1-i}}\prod^{N}_{n=1}\delta^{\strat{1}^n,\strat{2}}(\hat{\omega}_i,\hat{\omega}_{i+1}) \Expect{\starredstrat{1}}{\strat{2}}_{\StochG, \hat{\omega}_{N+1}}[\Rewards] \notag \\
    = \ & \lim_{N \rightarrow \infty} \sum_{\hat{\omega} \in v(V\setminus T)^{N+1}}\prod^{N}_{n=1}\delta^{\strat{1}^n,\strat{2}}(\hat{\omega}_i,\hat{\omega}_{i+1}) \Expect{\starredstrat{1}}{\strat{2}}_{\StochG, \hat{\omega}_{N+1}}[\Rewards] \label{th:semimarkov-to-memoryless-lim-eq2}\\
    \leq \ & \lim_{N \rightarrow \infty} \sum_{\hat{\omega} \in v(V\setminus T)^{N+1}}\prod^{N}_{n=1}\delta^{\strat{1}^n,\strat{2}}(\hat{\omega}_i,\hat{\omega}_{i+1}) \ M \label{th:semimarkov-to-memoryless-lim-eq3}\\
    = \ & M \ \lim_{N \rightarrow \infty} \sum_{\hat{\omega} \in v(V\setminus T)^{N+1}} \Prob{\strat{1}}{\strat{2}}_{\StochG,v}(\hat{\omega}) \label{th:semimarkov-to-memoryless-lim-eq4}\\
    = \ & 0  \label{th:semimarkov-to-memoryless-lim-eq5}
  \end{align}
  (\ref{th:semimarkov-to-memoryless-lim-eq1}) folows by noticing that,
  since terminal states in set $T$ are absorving,
  $V^{N+1}=(V\setminus T)^{N+1}\cup\bigcup_{i=0}^N(V\setminus T)^iT^{N+1-i}$.
  Since
  $\Expect{\starredstrat{1}}{\strat{2}}_{\StochG, \hat{\omega}_{N+1}}[\Rewards]=0$
  for all $\hat{\omega}_{N+1}\in T$, (\ref{th:semimarkov-to-memoryless-lim-eq2})
  folows.
  By taking
  $M = \max_{v\in V} \Expect{\starredstrat{1}}{\strat{2}}_{\StochG, v}[\Rewards]$,
  which we know it exists because of
  Theorem~\ref{th:memoryless-strat-p2-bounded-expectation}, we can conclude
  (\ref{th:semimarkov-to-memoryless-lim-eq3}).
  (\ref{th:semimarkov-to-memoryless-lim-eq4}) follows by definition of
  $\Prob{\strat{1}}{\strat{2}}_{\StochG,v}(\hat{\omega})$ and we
  conclude (\ref{th:semimarkov-to-memoryless-lim-eq5}) using
  Lemma~\ref{lemma:sum-of-nonterminal-is-zero}.
  

  Considering this last observation, from
  (\ref{th:semimarkov-to-memoryless-eq0}), we obtain
  \[
  \lim_{N \rightarrow \infty}	\sum^{N}_{n=0}\sum_{\hat{\omega} \in vV^{n}} \prod^{n-1}_{i=0}\delta^{\strat{1}^i,\strat{2}}(\hat{\omega}_i, \hat{\omega}_{i+1}) \reward(v_n) \ \leq \ \Expect{\starredstrat{1}}{\strat{2}}_{\StochG,v}[\Rewards]
  \]
  which  is equivalent to:
  \[
  \Expect{\strat{1}}{\strat{2}}_{\StochG,v}[\Rewards] \leq \Expect{\starredstrat{1}}{\strat{2}}_{\StochG,v}[\Rewards].
  \]
  Hence, the arbitrary semi-Markov strategy $\strat{1}$ does not improve on $\starredstrat{1}$.
  \qed
\end{myproof}

\end{document}